\emergencystretch=\maxdimen
\documentclass[paper=a4, fontsize=12pt]{scrartcl} 
\usepackage{xcolor}

\usepackage[utf8]{inputenc} 
\usepackage[english]{babel}
\usepackage{mathtools,amsmath,amsfonts,amsthm,amssymb} 
\newtheorem{definition}{Definition}
\newtheorem{proposition}[definition]{Proposition}
\newtheorem{lemma}[definition]{Lemma}
\newtheorem{remark}[definition]{Remark}
\newtheorem{corollary}[definition]{Corollary}
\newtheorem{theorem}[definition]{Theorem}
\numberwithin{equation}{section} 
\usepackage{authblk} 
\usepackage{indentfirst} 
\usepackage[margin=3cm]{geometry} 
\usepackage{hyperref}

\newcommand{\asl}{\widehat{\mathfrak{sl}_2}}
\newcommand{\aslp}{\widehat{\mathfrak{sl}_2^+}}
\newcommand{\aslm}{\widehat{\mathfrak{sl}_2^-}}
\newcommand{\Vir}{\mathsf{Vir}}
\newcommand{\U}{\mathrm{U}}
\newcommand{\C}{\mathbb C}
\newcommand{\N}{\mathbb N}
\newcommand{\Z}{\mathbb Z}

\newcommand{\V}{\mathcal V}
\renewcommand{\H}{\mathcal H}
\newcommand{\W}{\mathcal W}
\newcommand{\Wt}{\widetilde{\mathcal W} \vphantom{\mathcal W}}
\newcommand{\rad}{\mathrm{Rad}}
\newcommand{\w}{\mathfrak w}
\newcommand{\f}{\mathfrak f}
\newcommand{\F}{\mathcal F}

\newcommand{\psibar}{\overline \psi \vphantom{\psi}}
\newcommand{\chibar}{\overline \chi \vphantom{\chi}}

\newcommand{\J}{\mathcal J}
\newcommand{\sJ}{\mathsf J}

\newcommand{\ut}{{}^{\otimes}}

\newcommand{\Ker}{\mathrm{Ker}}
\newcommand{\Ran}{\mathrm{Ran}}

\title{\LARGE Decomposition of $\asl{}_{,k} \ \oplus \ \asl{}_{,1}$ highest weight representations for generic level $k$ and \\ equivalence between two-dimensional CFT models}

\author{Leszek Hadasz$^\sharp$\footnote{leszek.hadasz@uj.edu.pl, ORCID:0000-0002-8142-8185}\; and B\l{}a\.zej Ruba$^\dagger$}
\affil{\textit{\small$^\sharp$Institute of Theoretical Physics, Jagiellonian University in Kraków, Poland}}
\affil{\textit{\small$^\dagger$Department of Mathematical Sciences, University of Copenhagen, \protect\\
\vspace{-0.7em}
Universitetsparken 5, 2100 Copenhagen, Denmark }}

\date{\today}

\begin{document}
	
	\maketitle
	
\begin{abstract}
We construct highest weight vectors of $\asl_{,k+1} \oplus \Vir$ in tensor products of highest weight modules of $\asl_{,k}$ and $\asl_{,1}$, and thus for generic weights we find the decomposition of the tensor product into irreducibles of $\asl_{k+1} \oplus \Vir$. The construction uses Wakimoto representations of $\asl_{,k}$, but the obtained vectors can be mapped back to Verma modules. Singularities of this mapping are cancelled by a renormalization. A detailed study of ``degenerations'' of Wakimoto modules allowed to find the renormalization factor explicitly. The obtained result is a~significant step forward in a proof of equivalence of certain two-dimensional CFT models.
\end{abstract}

\newpage
\tableofcontents

\section{Introduction}
The Goddard-Kent-Olive (GKO) construction \cite{GKO} realizes minimal CFT$_2$ models as quotients, schematically\footnote{We would like to stress that equations \eqref{mrel0} and \eqref{eq:rel:intro:2} should not be viewed as relations between (product of) algebras, but as certain equivalences between conformal models with the indicated chiral symmetry structure. We believe that such of lack of precision is acceptable in the introductory section.} representing the space of states as
\begin{equation}
	\label{mrel0}
	V(k,m) \sim \frac{\asl{}_{,k}\times \asl{}_{,m}}{\asl{}_{,k+m}},
\end{equation}
where  $\asl{}_{,p}$ denote the $\asl$ Kac-Moody algebras at (integer) level $p.$
This result was used in \cite{Crnkovic:1989gy,Crnkovic:1989ug,Lashkevich:1992sb,Lashkevich:1993fb}
to derive a relation between the Virasoro minimal models $V(m)=V(1,m)$ and the ${\cal N}=1$  superconformal minimal models $SV(m)=V(2,m)$:
\begin{equation}
V(1)\otimes SV(m) \sim V(m)\otimes_P V(m+1), \;\;\;\;m=1,2,\dots,
\end{equation}
where the symbol $\otimes_P$ denotes a projected tensor product in which only selected pairs of conformal families are present.
The SL-LL relation of the form
\begin{equation}
	\label{eq:rel:intro:1}
	\hbox{free fermion}\,\otimes\,{\cal N} =1\ \hbox{super-Liouville}
	\hskip 8pt \sim \hskip 8pt
	\hbox{Liouville} \ \otimes_P\, \hbox{Liouville},
\end{equation}
can be seen as a non-rational counterpart of the relation (\ref{mrel0}). It was discussed in \cite{Bonelli:2011jx,Bonelli:2011kv,Belavin:2011sw} and explained  in \cite{Wyllard:2011mn}
in the context of the AGT correspondence \cite{Alday:2009aq}. 

The AGT relation was initially formulated for Virasoro conformal blocks on two dimensional Riemann surfaces and for SU(2) instanton partition functions in $\mathcal{N}=2$ supersymmetric gauge theories on $\mathbb{C}^2$ with the  $\Omega$ background \cite{Nekrasov:2002qd}. It was then generalized in various directions: the $\mathcal{N}=1$ super-Virasoro blocks turn out to correspond to SU(2) instanton partition functions on $\mathbb{C}^2/\mathbb{Z}_2$
\cite{Bonelli:2011jx, Bonelli:2011kv, Belavin:2011pp,  Belavin:2011tb,  Ito:2011mw, Belavin:2012aa},
conformal blocks of higher spin $\mathcal{W}_N$ algebra were related to SU$(N)$ instanton partition functions on $\mathbb{C}^2/\mathbb{Z}_N$ \cite{Wyllard:2009hg, Mironov:2009by}, and conformal blocks of the WZNW $\widehat{\mathfrak{sl}_n}$ level $k$ theory (with integer level $k$) turn out to be expressible through the SU$(k)$ instanton partition functions on $\mathbb{C}^2/\mathbb{Z}_n$ 
in the presence of surface operators \cite{Alday:2010vg,Foda:2019msm}.
As for the SL-LL relation itself, its proof  in the Neveu-Schwarz sector has been completed in \cite{SL-LL} while its  extension to the Ramond sector was analyzed in \cite{Schomerus:2012se}.

The relation
\begin{equation}
	\label{eq:rel:intro:2}
	\asl{}_{,k} \, \otimes\, \asl{}_{,1}
	\hskip 8pt \sim \hskip 8pt
	\hbox{Liouville}\, \otimes_P\, \asl{}_{,k+1}
\end{equation}
with an arbitrary level $k\in \mathbb{R}$ related to the central charge of the Liouville theory $c_{\scriptscriptstyle\mathrm{L}}$ as
\begin{equation}
\label{eq:Liouville:central:charge:through:k}
c_{\scriptscriptstyle\mathrm{L}} = 1-\frac{6}{(k+2)(k+3)},
\end{equation}
was put forward 
and analyzed in \cite{Jaskolski:2015uya}. It can be viewed as a non-trivial extension of the GKO coset construction of minimal models, with the  branching functions encoded in the projected tensor product.
Being  motivated by the SL-LL relation, the authors of \cite{Jaskolski:2015uya}
proposed to regard (\ref{eq:rel:intro:2}) as an exact equivalence of CFT$_2$ models. The proof of this equivalence was however missing in \cite{Jaskolski:2015uya}.

By a proof of the claim that the relation (\ref{eq:rel:intro:2}) is an equivalence between two quantum models we mean a demonstration that the corresponding correlation functions on both sides coincide. To achieve this one may start with equations (\ref{eq:J_curl}) and (\ref{eq:T_Vir}), which express generators of the 
$\mathsf{Vir}\, \oplus \, \asl{}_{,k+1}$ algebra (with the Virasoro algebra $\mathsf{Vir}$ being the chiral symmetry algebra of the Liouville theory) through generators of the $\asl{}_{,k} \, \oplus\, \asl{}_{,1}$ algebra and allow to consider representation spaces of the former as sub-spaces of representations spaces of the latter. Then one should:
\begin{enumerate}
\renewcommand{\theenumi}{\emph{\roman{enumi}}}
\item make the statement (\ref{eq:rel:intro:2}) precise by specifying which representations of the $\asl{}_{,k}\, \oplus\, \asl{}_{,1}$ we include in the spectrum of the theory, 
\item decompose these representations of $\asl{}_{,k}\, \oplus\, \asl{}_{,1}$ according to the action of $\mathsf{Vir}\, \oplus\, \asl{}_{,k+1}.$
\end{enumerate}
As for the point $i,$ in the present paper we are working with the highest weight representations of the involved algebras: Verma modules $\V^{k,j}\!$ of $\asl{}_{,k}$ (with general complex $k,j$) and the irreducible highest weight modules $\H^{1,\epsilon}$ of $\asl{}_{,1}$ (with $\epsilon =0$ or $\frac12$).

The decomposition we are after is discussed in Section \ref{section:tensor:product}. 
Our main findings relevant for the point $ii$ can be provisionally summarized as:
\begin{theorem} \label{thm:main_thm}
\begin{enumerate}
    \item For every $n \in \mathbb Z$ there exists a nonzero vector in $\V^{k,j} \otimes \H^{1,\epsilon} $ which is a~highest weight vector for $\mathsf{Vir}\, \oplus\, \asl{}_{,k+1}$ of conformal weight $\Delta_{k,j,\epsilon}^n$, given in \eqref{eq:highest_weights_Vir}, and $\asl{}_{,k+1}$ charge $j + \epsilon + n$;
    \item explicit form of this vector is given in Eq.\ \eqref{eq:special_states_Verma} (up to a possible re-normalization, see~Proposition \ref{prop:renormalization} and Corollary \ref{cor:existence:of:hws});
    \item for generic $k,j,$ $\V^{k,j} \otimes \H^{1,\epsilon}$ decomposes into a direct sum of Verma modules for $\mathsf{Vir}\, \oplus\, \asl{}_{,k+1}$ as in \eqref{eq:module:decomp}. 
\end{enumerate}
\end{theorem}

Point 3.\ of Theorem \ref{thm:main_thm} is known, see e.g.\ \cite{Bershtein:2013oka} (and \cite{Iohara:2002be} for a related statement for so-called admissible modules, which are not generic). The most important outcome of this work is an explicit expression \eqref{eq:special_states_Verma} for highest weight vectors $v^n_{\kappa,j,\epsilon}$, i.e.\  points 1.\  and 2.\  of Theorem \ref{thm:main_thm}. This construction is not restricted to generic $k,j$.



Now, in any two-dimensional CFT there exists a state-operator map. This map associates with a~highest weight vector $v_{k,j},$ c.f.\ eq.\ (\ref{eq:hw_conditions}), a primary field $\Phi( v_{k,j}|z).$ 
For the highest weight vectors defined by the $\asl$ level $k$ algebra (\ref{eq:sl2:algebra}) and for $z \in \mathbb{C},$  the primary field $\Phi( v_{k,j}|z)$ satisfies commutation relations:
\begin{equation}
\label{eq:primary:field:commutation:definition}
\left[\sJ^0_0,\Phi( v_{k,j}|z)\right] = j\Phi( v_{k,j}|z),
\hskip 1cm
\left[\sJ^+_0,\Phi( v_{k,j}|z)\right] = \left[\sJ^a_n,\Phi( v_{k,j}|z)\right] = 0,
\hskip 3mm n>0,\; a = \pm,0,
\end{equation}
as well as:
\begin{equation}
\left[L_{-1},\Phi( v_{k,j}|z)\right] = \partial_z\Phi( v_{k,j}|z),
\end{equation}
with $L_{-1}$ defined by (\ref{eq:T:expanded:in:modes}). Any other states are obtained by acting on $v_{k,j}$ with a symmetry algebra generators and the corresponding fields --- descendants of $\Phi( v_{k,j}|z)$ --- are defined as in \cite{Knizhnik:1984nr}.


The $z_i-$dependent factor of the three-point correlation function of primary fields\footnote{
To be precise, the $\sJ^0_0$ conservation law implies that the three-point correlation function of $\asl$ \emph{primary} fields vanish unless $j_1+j_2+j_3=0.$ To remedy this one replaces at least on of the fields $\Phi(v_{k,j}|z)$ with the ``primary field in the izospin basis'', $\Phi(x|z) = \Phi\left(\mathrm{e}^{x\sJ^-_0}v_{k,j}|z\right), x\in {\mathbb C}.$ } 
$\Phi(v_{k,j_i}|z_i)$, $i=1,2,3,$ is universal (its form follows from the global conformal invariance) and the essential information this function carries has the form of the three-point structure \emph{constant}, a function of the level $k$ and weights $j_i.$ Structure constants of the Liouville theory were computed  in \cite{Dorn:1994xn,Zamolodchikov:1995aa,Teschner:1995yf,Teschner:2001rv} and the structure constants of the $H^+_3 = \mathrm{SL}(2,\mathbb{C})/\mathrm{SU}(2)$ coset model (with its underlying $\asl$ chiral symmetry algebra) were determined by the bootstrap method in  \cite{Gawedzki:1991yu,Teschner:1997fv,Teschner:1997ft,Teschner:1999ug}  and by path integral methods in \cite{Ishibashi:2000fn,Hosomichi:2000bm,Hosomichi:2001fm,Satoh:2001bi}\footnote{On the other side of the $c=1,$ resp. $k=-2$ ``bareer'' we  have, in the case of the Liouville theory, the so called ``generalized minimal models'' 
\cite{Zamolodchikov:2005fy,Kostov:2005kk,Ribault:2015sxa} 
and in the case of the $H^+_3$ model the (real) SU$(2)$ WZNW model,
\cite{Gawedzki:1988hq,Gawedzki:1988nj,Gawedzki:1989rr}.}


Formula (\ref{eq:special_states_Verma}) for the $\mathsf{Vir}\,\oplus\,\asl{}_{,k+1}$  highest weight vectors $v^n_{\kappa,j,\epsilon}$ 
expresses them  explicitly as  ``excited'' states in the highest weight representation space of  the $\asl{}_{,k} \, \oplus\, \asl{}_{,1}$ algebra.  The corresponding (via state - operator map) fields are  descendants of the $\asl{}_{,k} \, \oplus\, \asl{}_{,1}$ primary fields.
Therefore, when computing three-point correlation functions of fields $\Phi\big(v^{n}_{\kappa,j,\epsilon}\big|z\big),$ we can treat them either as primary fields with respect to the $\mathsf{Vir}\, \oplus\, \asl{}_{,k+1}$
algebra and express the correlator in question directly through the product of the Liouville and $H^+_3/$WZNW (level $k+1$) structure constants, or use the $\asl$ Ward identities and express the very same three-point correlation function through the product of the corresponding three-point $\asl$ blocks and the $H^+_3/$WZNW  structure constants (at levels 1 and $k$). 

\newpage
The third step in the proof of the equivalence statement then consists in:
\begin{enumerate}
\renewcommand{\theenumi}{\emph{\roman{enumi}}}
\setcounter{enumi}{2}
\item computing the ratio of three-point correlation function
\begin{equation}
\label{eq:three:point:ratios}
\left\langle
\Phi\big(v^{n_3}_{\kappa,j_3,\epsilon_3}\big|z_3\big)
\Phi\big(v^{n_2}_{\kappa,j_2,\epsilon_2}\big|z_2\big)
\Phi\big(v^{n_1}_{\kappa,j_1,\epsilon_1}\big|z_2\big)
\right\rangle
\end{equation}
and the product of structure constants of the Liouville and $H^+_3/$WZNW  theories in two ways and comparing the results.
\end{enumerate}
The ratio of the corresponding three-point constants was calculated in \cite{Jaskolski:2015uya}, but only the construction of vectors $v^n_{\kappa,j,\epsilon}$ allows to compute the relevant three-point block. As expected from the results of \cite{Jaskolski:2015uya}, this block can be presented as a product of factors being affine functions of $j_1, j_2$ and $j_3.$ 
We shall present the full derivation of the form of the three-point block in question, using the free field methods presented here, in the forthcoming publication \cite{HR2024_to_come}.

Since in a conformal field theory the higher point correlation functions are uniquely determined by the three-point structure constants and the chiral symmetry algebra,  completion of tasks $i - iii$ will essentially complete the program initiated in \cite{Jaskolski:2015uya}\footnote{We leave aside a question of modular invariance of the related theories; based on the result of \cite{SL-LL} and \cite{Jaskolski:2015uya} we expect that the spectrum of the theory on the r.h.s.\ of (\ref{eq:rel:intro:2}) will be non-diagonal, with the relatively shifted scaling dimension of chiral and anti-chiral components, and the modular invariance of such a theory need to be separately checked.}.

Let us now discuss the content of the present paper in more detail.
As a preparation to study tensor products, in Section \ref{sec:prelim} we review the relevant elements of representation theory of $\asl$. Most of Section \ref{sec:prelim} consists of known results. We hope that it will help some readers in navigating the (rather technical) existing mathematical literature on the subject. We would like to highlight three parts of Section \ref{sec:prelim}. Subsection \ref{sec:sing_vec} is devoted to singular vectors in Verma modules of $\asl$, in particular the explicit construction of singular vectors due to Malikov, Feigin and Fuchs (MFF) \cite{MFF}. The formula obtained in \cite{MFF} involves complex powers of Lie algebra elements. A rigorous interpretation of such expressions was given in \cite{MFF}, but we found it to be not very convenient. We present our own approach based on noncommutative localization. The required background material is included in Appendix \ref{app:localization}. In Subsection \ref{sec:norm} we state a~conjecture \eqref{eq:singular:vector:norms:result} on normalizations factors of singular vectors. Subsection \ref{sec:fermi} is devoted to representations of $\asl$ at level $k=1$ constructed in terms of free fermions.

In order to analyze representations of $\asl$ with general $k$ we consider Wakimoto representations, which are introduced in Section \ref{sec:Wakimoto}. In the Wakimoto representations elements of $\asl$ are expressed in terms of three types of free field operators. One of the $\asl$ currents is linear in free fields, one is quadratic and one is cubic. The fact that some generators of the algebra become linear in free fields is responsible for significant simplifications in some calculations, for example one of the three free fields involved in the construction of the Wakimoto representation does not appear in expressions for singular vectors. After introducing two types of Wakimoto representations, in Subsection \ref{sec:bases} we study their degenerations, or equivalently, singularities of the transformation between $\sJ$ operators and free fields. In particular we show that the celebrated determinant of Kac and Kazhdan \cite{KK} factors into two (explicit) expressions which control singularities of the two Wakimoto modules. One consequence of this factorization is that in each type of Wakimoto modules one sees only ``half'' of singular vectors. In Subsection \ref{sec:screenieng_charges} we discuss how these singular vectors can be constructed using screening charges. This construction is particularly useful in the limit $k \to \infty$, where we obtain simple asymptotics. By contrast, structure of vectors obtained in the MFF construction is easier to understand for small $k$.

Section \ref{section:tensor:product} contains the construction of the highest weight vectors $v^n_{\kappa,j,\epsilon}$. We remark that in order to define $v^n_{\kappa,j,\epsilon}$ we need the Wakimoto representations, but we could do without the free fermion representation. However, the free fermion representation was crucial in checking that $v^n_{\kappa,j,\epsilon}$ satisfy highest weight equations. Moreover, we needed to include certain explicit normalization factors in $v^n_{\kappa,j,\epsilon}$ to cancel singularities of the transformation from free fields to $\sJ$ operators. The calculation which allowed to find these normalization factors relies on the detailed understanding of singular vectors in Wakimoto representations. Having found the highest weight vectors, we finally deduce Theorem \ref{th:character:decomposition}, which, for generic weights, gives the decomposition of the tensor product of Verma modules for $\asl{}_{,k}$ and $\asl{}_{,1}$ into irreducible representations of $\asl{}_{,k+1}$ and $\Vir$ and implies a corresponding decomposition of characters. 

Let us finally notice that since Wakimoto representation, being in a technical heart of the present work, is known also for higher rank affine algebras, it is conceivable that results presented here can be extended beyond rank 2 case.

\textbf{Note added.} A few months after this work was made public in the ArXiv database there appeared a paper 
\cite{Bershtein:2024kwe},
which contains results (sometimes expressed in a different representation) partially overlapping with ours and partially extending them. 




\section{Preliminaries on \texorpdfstring{$\asl$}{affine sl2}} \label{sec:prelim}
\subsection{Representations and currents} \label{sec:currents}

$\asl$ is a complex Lie algebra with basis $\{ {\mathsf J}^{\pm}_n, {\mathsf J}^0_n \}_{n \in \mathbb Z} \cup \{ \mathsf K \}$. The nonzero commutators of basis elements are:
\begin{equation}
\label{eq:sl2:algebra}
[{\mathsf J}_n^0,{\mathsf J}_m^{\pm}] = \pm {\mathsf J}_{n+m}^{\pm}, \quad [{\mathsf J}_n^+,{\mathsf J}_m^-] = 2 {\mathsf J}_{n+m}^0 +  n\delta_{n+m}  \mathsf K, \quad [{\mathsf J}_n^0,{\mathsf J}_m^0] = \frac{n}{2} \delta_{n+m} \mathsf K.
\end{equation}
It is sometimes convenient to write these relations as $[{\mathsf J}_n^a, {\mathsf J}_m^b] = \sum \limits_c f^{ab}_{\ \ c} {\mathsf J}_{n+m}^c + n \delta_{n+m} h^{ab} \mathsf K$, with self-evident $f^{ab}_{\ \ c}$ and $h^{ab}$ tensors. 

All representations $V$ of $\asl$ considered in this paper have the following properties:
\begin{enumerate}
    \item \label{ass:field} For every $v \in V$ there exists $n$ such that $\mathsf J^a_m v =0$ for $m \geq n$,
    \item \label{ass:level} $\mathsf K$ acts as a scalar $k \neq -2$, called the level.
\end{enumerate}

Currents $\mathsf J^a(z)$ are defined as the following operator-valued formal series:
\begin{equation}
    {\mathsf J}^a(z) = \sum_{n \in \Z} \frac{{\mathsf J}^a_n}{z^{n+1}}.
    \label{eq:current_modes}
\end{equation}
Assumption \ref{ass:field} is equivalent to the statement that for every $v \in V$, the series ${\mathsf J}^a(z)v$ has finitely many nonzero terms with negative powers of $z$. This is the definition of a field e.g.~in \cite{Vertex}. 

Two fields $A,B$ are said to be local with respect to each other if\footnote{For fermionic fields, which will be used later, this formula is modified by a sign.} for some integer $N$ 
\begin{equation}
(z-w)^NA(z)B(w)=(z-w)^N B(w) A(z).
\end{equation}
This simple assumption is enough to prove existence of operator product expansions (OPE) 
\begin{equation}
    A(z) B(w) \sim \frac{C_n(w)}{(z-w)^N} + \dots + \frac{C_1(w)}{z-w},
\end{equation}
where $\sim$ denotes expansion for $z \to w$ with regular terms omitted. The order $(z-w)^0$ term in the expansion is denoted $:\!A(w)B(w)\!:$ and called the normal (or~normally ordered) product of $A$ and $B$. If the product $A(z)B(w)$ is nonsingular for $z \to w$ (so $A(z)B(z)$ makes sense and equals $:\!A(z)B(z)\!:$), we write
\begin{equation}
    A(z) B(w) \sim \mathrm{reg}.
\end{equation}
This is equivalent to $A(z)B(w)=B(w) A(z)$.

Commutation relations of operators $\mathsf J^a_n$ are equivalent to the statement that fields $\mathsf J^a$ are local and satisfy the OPE
\begin{equation}
    {\mathsf J}^a(z) {\mathsf J}^b(w) \sim \frac{1}{z-w} \sum_c f^{ab}_{\ \ c} {\mathsf J}^c(w) + \frac{h^{ab} k}{(z-w)^2}.
    \label{eq:current_OPE}
\end{equation}

Sugawara's construction of the Virasoro field gives
\begin{equation}
    T(z) = \frac{1}{k+2} \left( :{\mathsf J}^0(z) {\mathsf J}^0(z): + \frac12 :{\mathsf J}^+(z) {\mathsf J}^-(z): + \frac12 :{\mathsf J}^-(z) {\mathsf J}^+(z) \right).
    \label{eq:Sugawara}
\end{equation}
$T$ satisfies the following OPEs:
\begin{gather}
T(z) T(w)  \sim \frac{\partial T(w)}{z-w} + \frac{2 T(w)}{(z-w)^2} + \frac{c_k}{2 (z-w)^4}, \quad \text{with } c_k = \frac{3k}{k+2}, \\
    T(z) {\mathsf J}^a(w)  \sim \frac{\partial {\mathsf J}^a(w)}{z-w} + \frac{{\mathsf J}^a(w)}{(z-w)^2}. \nonumber
\end{gather}
This implies that operators $L_n$ in the series expansion
\begin{equation}
\label{eq:T:expanded:in:modes}
    T(z) = \sum_{n \in \Z} \frac{L_n}{z^{n+2}}
\end{equation}
satisfy the commutations relations: 
\begin{equation}
 [L_n, L_m] = (n-m) L_{n+m} + \frac{c_k}{12} n (n^2-1) \delta_{n+m,0}, \qquad  [L_n , {\mathsf J}^a_m ] = - m {\mathsf J}^a_{n+m}.
\end{equation}
In particular, $V$ is a representation of the Virasoro algebra with central charge $c_k$.

\begin{lemma} \label{lem:Virasoro_invariant_submodules}
    If $V' \subset V$ is a submodule, i.e.\ it is mapped to itself by every $\mathsf J^a_n$, then $L_n V \subset V$ for every $n$ as well.
\end{lemma}
\begin{proof}
Assumption \ref{ass:field} implies that for every $v \in V$, $L_n v=Tv$ for some operator $T$ which is a~second order polynomial in finitely many $\mathsf J^a_n$ (depending on $v$).    
\end{proof}

We are now ready to state another crucial property of representations we study:
\begin{enumerate}
 \setcounter{enumi}{2}
    \item \label{ass:diag} Operators $L_0$ and $\mathsf J_0^0$ are diagonalizable.
\end{enumerate}
Since $L_0$ commutes with $\mathsf J_0^0$, assumption \ref{ass:diag} implies
\begin{equation}
V= \bigoplus \limits_{\Delta, j \in \C} V_{\Delta, j}, \qquad V_{\Delta, j} = \{ v \in V  | \, L_0 v = \Delta v, \ {\mathsf J}_0^0 v = jv \}.
\label{eq:weight_decomp}
\end{equation}

Lemma \ref{lem:Virasoro_invariant_submodules} shows that if $V$ satisfies assumptions 1-3, then every submodule $V'$ and every quotient $V/V'$ also satisfies these assumptions, and we have
\begin{equation}
    V'_{\Delta,j} = V' \cap V_{\Delta,j}, \qquad (V/V')_{\Delta,j}= V_{\Delta,j} / V'_{\Delta,j}.
\end{equation}

A bilinear form $(\cdot, \cdot)$ on $V$ is said to be contragradient if for all $v,v' \in V$ we have
\begin{equation}
(v, \mathsf J^a_n v') = (\mathsf J^{-a}_{-n} v,v') \qquad \forall v,v' \in V.
\end{equation}
This implies that we also have $(v,L_n v')= (L_{-n} v,v')$. By specializing to $n=0$ in both equations, we see that $(v,v')=0$ if $v \in V_{\Delta,j}$, $v' \in V_{\Delta',j'}$ and $(\Delta,j) \neq (\Delta',j')$. 

The contragradient dual of $V$ is defined as
\begin{equation}
V^{\vee} = \bigoplus_{\Delta,j} V_{\Delta,j}^{\vee},
\end{equation}
where $V_{\Delta,j}^{\vee}$ is the dual space of $V_{\Delta,j}$. Action of $\asl$ on $V^{\vee}$ is defined by
\begin{equation}
({\mathsf J}^a_n \varphi)(v) = \varphi ({\mathsf J}^{-a}_{-n}v) \qquad \mathrm{for} \ \varphi \in V^{\vee}, \, v \in V.
\end{equation}
Assumptions 1-3 for $V$ imply that the same is true for $V^\vee$. A contragradient bilinear form on $V$ is the same as a module homomorphism $V \to V^\vee$.

\begin{remark}
If $k = -2$, it is not possible to define an operator $L_0$ using the construction \eqref{eq:Sugawara}. In some cases it is still possible to define a diagonalizable operator $\delta$ such that $[\delta,\mathsf J^a_n]= -n \mathsf J^a_n$, and hence to define a decomposition analogous to \eqref{eq:weight_decomp}. Asking for existence of such $\delta$ amounts to studying representations of the Lie algebra $\asl$ with $\delta$ suplemented as an independent generator. This enlarged algebra is preferred in much of the literature.
\end{remark}

\subsection{Highest weight modules}

We let $\aslp$ (resp. $\aslm$) be the subalgebra of $\asl$ generated by ${\mathsf J}_0^+$ and $\{ {\mathsf J}_n^a \}_{n \geq 1}$ (resp.~${\mathsf J}_0^-$ and $\{ {\mathsf J}_n^a \}_{n \leq -1}$). A nonzero element $v \in V$ is said to be a highest weight vector if it satisfies
\begin{equation}
 \mathsf K v = k v, \qquad  {\mathsf J}_0^0 v = j v, \qquad \aslp v =0.
 \label{eq:hw_conditions}
\end{equation}
The pair $(k,j) \in \C^2$ (or just $j$, if $k$ is clear from the context) is called the weight of~$v$. Conditions \eqref{eq:hw_conditions} imply:
\begin{equation}
L_0 v = \Delta_{k,j} v, \ \ \text{where} \ \ \Delta_{k,j}= \frac{j(j+1)}{k+2}, \qquad L_n v =0 \ \ \text{for} \ \ n > 0.
\end{equation}

Representation $V$ is said to be a highest weight module if there exists a highest weight vector $v \in V$ which is cyclic. Such $v$ is unique up to multiples, see Proposition~\ref{prop:hw_weights}. Its weight is called the highest weight of $V$. 

\begin{proposition}
\label{prop:hw_weights}
Let $k \neq -2$ and let $V$ be a highest weight module with highest weight $(k,j).$ Assumptions 1-3 in Subsection \ref{sec:currents} are satisfied. Subspaces $V_{\Delta,j}$ are finite dimensional. $V_{\Delta,j}=0$ for all $(\Delta,j)$ which are not of the form
\begin{equation}
    [n,m]= (\Delta_{k,j}+n,j+n-m), \qquad n,m \in \N.
    \label{eq:weights_grading}
\end{equation}
We have $\dim(V_{[0,0]})=1$. Nonzero elements of $V_{[0,0]}$ are the only cyclic highest weight vectors of~$V$. There exists a proper submodule $\rad(V) \subsetneq V$ containing all proper submodules of $V$. $V / \rad(V)$ is irreducible.
\end{proposition}
\begin{proof}
Let $v$ be a highest weight vector of $V$. By Poincar\'e-Birkhoff-Witt theorem, $V$ is spanned by elements $v' = \prod_{i=1}^N \mathsf J^{a_i}_{-n_i} \cdot v$ with each $J^{a_i}_{-n_i} \in \aslm$. Such $v'$ satisfies $\mathsf K v' = k v'$ and is annihilated by $\mathsf J^a_n$ if $n > \sum_{i=1}^N n_i$. We have $v' \in V_{[n,m]}$, with
\begin{equation}
    n = \sum_{i=1}^N n_i , \qquad m = \sum_{i=1}^N (n_i-a_i). 
    \label{eq:nm_defi}
\end{equation}
$n,m$ are natural numbers because each $n_i$ and each $n_i-a_i$ is a natural number. There exists finitely many sequences of $n_i$ and $a_i$ such that \eqref{eq:weight_decomp} holds for a given $n,m$, so $\dim(V_{[n,m]}) < \infty$. If $n=m=0$, the only such sequence is the empty sequence, corresponding to $v'=v$. 

Now let $w$ be another highest weight vector of $V$. Then $w \in V_{[n,m]}$ for some $n,m \in \N$. Since $w$ is cyclic, $V_{[n',m']}=0$ if $n' < n$ or $m' < m$. Thus $n=m=0$, so $w $ is proportional to $v$.

Suppose that $V' \subsetneq V$ is a submodule. Then $v \not \in V'$, so $V'_{[0,0]}=0$ by Lemma~\ref{lem:Virasoro_invariant_submodules}. This shows that the sum of proper submodules of $V$ is a proper submodule. Let $\rad(V)$ be the union (or sum) of all proper submodules of $V$. $\rad(V)$ is a~submodule and $\rad(V)_{[0,0]}=0$, so $\rad(V) \neq V$.
\end{proof}

Verma module $\V^{k,j}$ is defined \cite{Verma} as the quotient of the universal enveloping algebra of $\asl$ by the left ideal generated by $\aslp \cup \{ \mathsf K -k , {\mathsf J}_0^0 - j \}$. We denote the image of $1$ in $\V^{k,j}$ by $v_{k,j}$. It is a~cyclic highest weight vector of weight $(k,j)$.

The following properties of $\V^{k,j}$ are well-known.
\begin{enumerate}
\item If $v \in V$ is a highest weight vector of weight $(k,j)$, there exists a~unique homomorphism $\V^{k,j} \to V$ taking $v_{k,j}$ to $v$. Clearly it is surjective if $v$ is cyclic. Therefore, every highest weight module is a quotient of a Verma module. In~particular, the irreducible highest weight modules are $\H^{k,j}= \V^{k,j} / \rad(\V^{k,j})$. 
\item Every nonzero element of $\U \aslm$ acts as an injective endomorphism of $\V^{k,j}$. In particular every nonzero homomorphism $\V^{k,j'} \to \V^{k,j}$ is injective. This propert follows from the Poincar\'e-Birkoff-Witt theorem and the fact that the product of two nonzero elements of a universal enveloping algebra is nonzero \cite[ch.~2.3]{Dixmier}.  
\item There exists a nonzero contragradient bilinear form $g$ on $\V^{k,j}$, unique up to constant multiples. Its normalization may be fixed by setting $g(v_{k,j},v_{k,j})=1$. $g$ is symmetric. Existence and uniqueness of $g$ may be proven as in \cite{Shapovalov} for Verma modules of semisimple Lie algebras.
\end{enumerate}

\begin{proposition} \label{prop:gker}
Let $k \neq -2$. We have $\rad(\V^{k,j}) = \Ker(g)$. Therefore, every highest weight module $V$ admits a unique (up to multiples) nonzero contragradient bilinear form $g$. $g$ is symmetric, $g(v,v) \neq 0$ for the cyclic highest weight vector $v$ of $V$, and the kernel of $g$ is $\rad(V)$.
\end{proposition}
\begin{proof}
$\Ker(g) \subset \rad(\V^{k,j})$ is clear (for any $k$). For $\supset$, consider $v \not \in \Ker(g)$ and let $V$ be the $\asl$-module generated by $v$. There exists $T \in \U \asl$ such that $g(T v_{k,j}, v) \neq 0$, so there exists $T' \in \U \asl$ such that $g(v_{k,j}, T'v) \neq 0$. Then $T' v \in V$ has a nonzero component in $\V^{k,j}_{[0,0]}$. By Lemma \ref{lem:Virasoro_invariant_submodules}, the projection of $T'v$ onto $\V^{k,j}_{[0,0]}$ is also in $V$, so~$v_{k,j} \in V$ and hence $v \not \in \rad(\V^{k,j})$.

Now let $V$ be a highest weight module. There exists a surjection $\V^{k,j} \to V$ whose kernel is contained in the kernel of $g$ on $\V^{k,j}$, so there is an induced $g$ on $V$. Its claimed properties follow from the corresponding properties of $g$ on $\V^{k,j}$. Now let $g'$ be another nonzero contragradient bilinear form on $V$ and let $v$ be a cyclic highest weight vector of $V$. Put $h = g' - \frac{g'(v,v)}{g(v,v)} g$. $h$ is contragradient and $h(v,v)=0$, from which one easily deduces $h=0$.
\end{proof}

\begin{remark} \label{rmk:renormalized}
There exist operators $\delta$ and $(L_n)_{n=1}^\infty$ on $\V^{-2,j}$ uniquely determined by 
\begin{equation}
\delta v_{-2,j} = L_n v_{-2,j} =0, \quad [\delta,\mathsf J^a_m] = -m \mathsf J^a_m, \quad [L_n, \mathsf J^a_m] = -m \mathsf J^a_{n+m}. 
\end{equation}
Uniqueness: these rules clearly suffice to evaluate $\delta$ and $L_n$ on any element of $\V^{-2,j}$. \\ Existence: consider the derivation of $\asl$ defined by $\mathsf J^a_m \mapsto -m \mathsf J^a_{n+m}, \mathsf K \to 0$. If $n \geq 0$, its extension to the universal enveloping algebra preserves the left ideal generated by $\aslp$, $\mathsf K-k$ and $\mathsf J^0_0 -j$, so it descends to a linear map on the quotient space.

If $k\neq -2$, operators $L_n$ constructed above coincide with those obtained from~\eqref{eq:Sugawara}, and we have $\delta = L_0 - \Delta_{k,j}$. Thus $\delta$ may be thought of as $L_0$ renormalized by subtracting a divergent constant. To obtain finite limits of $L_{n}$ with $n<0$, a multiplicative renormalization is required\footnote{Here we take the limit of matrix elements relative to the canonical basis, see Subsection \ref{sec:KK_determinant}.}: 
\begin{equation}
L_n^{\mathrm{ren}}= \lim_{k \to -2} (k+2) L_n.
\end{equation}
This renormalization changes the commutation rules. In particular, $L_n^{\mathrm{ren}}$ commute with each other and with all $\mathsf J^a_m$. 

If $k=-2$, operators $\delta$ and $L_n$ ($n>0$) can not be expressed in terms of $\mathsf J^a_n$, so the proof of Lemma \ref{lem:Virasoro_invariant_submodules} is not applicable. Therefore, the argument showing that the union of all proper submodules of $\V^{k,j}$ is a proper submodule does not extend to $k=-2$. As for a counterexample, the vector $L_{-1}^{\mathrm{ren}} v_{-2,j}$ is annihilated by $\aslp$, so the $\asl$-module it generates does not contain $v_{-2,j}$. On the other hand, $L_1 L_{-1}^{\mathrm{ren}} v_{-2,j} = 2j(j+1) v_{-2,j}$. It is also easy to see that $v_{-2,j}$ is not in the $\asl$-module generated by $(1+L_{-1}^{\mathrm{ren}}) v_{-2,j}$, even though $v_{-2,j}$ is a homogeneous of component of $(1+L_{-1}^{\mathrm{ren}})v_{-2,j}$. Moreover, the two vectors $(1 \pm L_{-1}^{\mathrm{ren}})v_{-2,j}$ separately generate proper submodules, but the sum of these submodules is the whole $\V^{-2,j}$.

One may define $\rad(\V^{-2,j})$ to be the union of all $\delta$-invariant proper submodules. Then, by the proof of Proposition \ref{prop:gker}, $\rad(\V^{-2,j})=\Ker(g)$, and the quotient module $\V^{-2,j}/\rad(\V^{-2,j})$ has no nontrivial $\delta$-invariant submodules. Operators $L_n^{\mathrm{ren}}$ ($n<0$) take whole $\V^{-2,j}$ to $\rad(\V^{-2,j})$, but $\rad(\V^{-2,j})$ is not $L_n$-invariant ($n>0$). In \cite{Malikov} it was shown that if $2j+1 \not \in \Z \setminus \{ 0 \}$, the space of highest weight vectors in $\V^{-2,j}_{[n,n]}$ is spanned by vectors $L_{-i_1}^{\mathrm{ren}} \cdots L_{-i_p}^{\mathrm{ren}} v_{-2,j}$ with $\sum i_p =n$, and vectors $L_{-i}^{\mathrm{ren}}v_{-2,j}$ generate the module $\rad(\V^{-2,j})$.   
\end{remark}

\subsection{Kac-Kazhdan determinant} \label{sec:KK_determinant}

Consider monomials of the form
\begin{equation}
\label{eq:multiindex:J}
{\mathsf J}_{-I} = \left[ \prod_{l=1}^{\infty} \prod_{a = -1}^1 \left( {\mathsf J}_{-l}^a \right)^{i_{l,a}} \right] \left( {\mathsf J}_0^- \right)^{i_{0,-}} = \cdots ({\mathsf J}_{-1}^-)^{i_{1,-}} ({\mathsf J}_{-1}^0)^{i_{1,0}} ({\mathsf J}_{-1}^{+})^{i_{1,+}} ({\mathsf J}_0^-)^{i_{0,-}}, 
\end{equation}
where $I=(i_{0,-},i_{1,+},i_{1,0},i_{1,-},...)$ is a multi-index whose entries are natural numbers, all but finitely many of which are zero. We shall write $I \in \mathcal P_{n,m}$ if
\begin{equation}
n = \sum_{l,a} i_{l,a} l, \qquad m = \sum_{l,a} i_{l,a} (l-a).
\end{equation}
$\{ {\mathsf J}_{-I} v_{k,j} \}_{I \in \mathcal P_{n,m}}$ is a basis of $\V^{k,j}_{[n,m]}$, called the canonical basis. Let $d_{n,m} =\dim(\V^{k,j}_{[n,m]})$.

A family of vectors in $\V^{k,j}_{[n,m]}$ depending on parameters $k,j$ is said to be polynomial (resp.\ rational) if its coefficients relative to the canonical basis are polynomials (resp.~rational functions) in $k,j$. Similar conventions are adopted for bilinear forms, tensors etc.

Let $g_{n,m}$ be the restriction of $g$ to $\V^{k,j}_{[n,m]}$. Then $g_{n,m}$ is polynomial in $k,j$. Its determinant relative to the canonical basis was computed in \cite{KK} (in~fact for Verma modules of a~certain class of Lie algebras including $\asl$), up to a~constant which is easy to deduce by analyzing the proof therein. 

To state the result, it is convenient to introduce the following notations. Let $\mathcal S^+ = \N \times (\N \setminus \{ 0 \})$, $\mathcal S^- = (\N \setminus \{ 0 \})^2$. Then for $a \in \{ \pm 1 \}$ and $(r,l) \in \mathcal S^a$ we put:
\begin{equation}
    T_{r,l}^a(k,j) = r (k+2) -l+ a(2j+1).
\end{equation}

\begin{theorem} \label{thm:KDet}
Relative to the canonical basis we have
\begin{equation}
 \det(g_{n,m})  =  \prod_{I \in \mathcal P_{n,m}} I! \cdot \left( \frac{k+2}{2} \right)^{\sum \limits_{r,l=1}^{\infty} d_{n-rl,m-rl}} \cdot \prod_{\substack{a \in \pm 1 \\ (r,l) \in \mathcal S^a }} T_{r,l}^a(k,j)^{d_{n-rl,m-(r+a)l}}.
 \label{eq:Kac_det}
\end{equation}
\end{theorem}

The following Corollary of \eqref{eq:Kac_det} was also noted in \cite{KK}.

\begin{corollary} \label{cor:reducibility_dim1}
$\mathcal V^{k,j} $ is a reducible module if and only if $k \neq -2$ or there exist $a \in \{ \pm \}$ and $(r,l) \in \mathcal S^{a}$ such that $T^a_{r,l}(k,j) = 0$. If $k \neq -2$ and there exists exactly one such triple $(a,r,l)$, then $\rad(\V^{k,j}) \cong \V^{k, j-al}$.
\end{corollary}
\begin{proof}
   If $k = -2$, then $\Ker(g)$ is a nonzero proper submodule of $\V^{-2,j}$. Characterization of reducibility for $k \neq -2$ follows from Proposition \ref{prop:gker}. Now suppose that $k \neq -2$ and exactly one $T^a_{r,l}(k,j)$ vanishes. Then $\det(g_{rl,(r+a)l})$ has a~simple zero at $k,j$, so~$g_{rl,(r+a)l}$ has one-dimensional kernel. Let $u \in \Ker(g_{rl,(r+a)l})$ be nonzero. Then $\sJ^b_n u \in  \Ker(g_{rl-n,(r+a)l-n-b})$, which is the zero subspace for $n >0$ and for $(n,b)=(0,+)$. Hence $u$ is a highest weight vector, so the submodule $M$ generated by $u$ is isomorphic to $\V^{k,j-al}$. For every $n,m$ we have $M_{[n,m]} \subset \Ker(g_{n,m})$, and $\dim M_{[n,m]}$ equals the order of the zero of $\det(g_{n,m})$, so $M_{[n,m]} = \Ker(g_{n,m})$.
\end{proof}

\subsection{Singular vectors} \label{sec:sing_vec}

A non-cyclic highest weight vector in a highest weight module is called a singular vector. Every singular vector belongs to $\Ker(g)$. Singular vectors in Verma modules are in one-to-one correspondence with embeddings of Verma modules.

First we discuss an existence result for singular vectors (cf.~Remark \ref{rmk:renormalized} for a partial discussion of an additional construction for $k+2=0$). If $T^{a}_{r,l}(k,j)=0$, then there exists a singular vector in $\V^{k,j}_{[rl,(r+a)l]}$ (that is, an embedding $\V^{k,j-al} \hookrightarrow \V^{k,j}$). For generic $k,j$ this follows from Corollary~\ref{cor:reducibility_dim1}. The genericness assumption may be removed because the function
\begin{equation}
    (k,j) \mapsto \dim \{ v \in \V^{k,j} \, | \, \sJ_0^+ v = \sJ_1^- v =0 \}
\end{equation}
is upper semicontinuous. 

One may obtain new examples of singular vectors by composing the corresponding embeddings of Verma modules. As an example, the result in the previous paragraph shows existence of singular vectors in $\V^{0,0}_{[0,1]}$ and $\V^{0,-1}_{[3,0]}$, or equivalently two embeddings
\begin{equation}
    \V^{0,2} \hookrightarrow \V^{0,-1} \hookrightarrow \V^{0,0}.
\end{equation}
Their composition is an embedding $\V^{0,2} \hookrightarrow \V^{0,0}$ corresponding to a singular vector in $\V^{0,0}_{[3,1]}$. We note that $[3,1]$ is not equal $[rl,(r+a)l]$ for any labels $a,r,l$ corresponding to zeros of the Kac-Kazhdan determinant. 

Next we discuss to what extent the above existence results are exhaustive. Theorem $2$ in \cite{KK} shows that\footnote{In fact Theorem $2$ in \cite{KK} is stronger. It characterizes irreducible subquotients of Verma modules. Since every singular vector gives rise to an irreducible subquotient, it implies the statement discussed here.} if there exists a singular vector $v \in \V^{k,j}$ of weight $j'$ (embedding $\V^{k,j'} \hookrightarrow \V^{k,j}$), then there exist $s \geq 1$ and a sequence of triples $\{ (a_i,r_i,l_i) \}_{i=1}^s$ such that, defining $j_0=j$ and $j_i = j_{i-1} - a_i l_i$ for $i \in \{ 1 , \dots, s \}$, we have $j_s =j'$ and $T^{a_i}_{r_i,l_i}(k,j_{i-1})=0$ for each $i \in \{ 1 , \dots , s \}$. Thus there exists a~sequence of embeddings
\begin{equation}
    \V^{k,j'} = \V^{k,j_s} \hookrightarrow \V^{k,j_{s-1}} \hookrightarrow \dots \hookrightarrow \V^{k,j_{0}} = \V^{k,j}.
    \label{eq:embedding_resolved}
\end{equation}
In this sense the existence of an embedding $\V^{k,j'} \hookrightarrow \V^{k,j}$ could be concluded from Corollary~\ref{cor:reducibility_dim1}. We emphasize that this by itself is not saying that every embedding $\V^{k,j'} \hookrightarrow \V^{k,j}$ is equal to a~composition as in \eqref{eq:embedding_resolved}. Thus one is lead to the question what is the dimension of the space of singular vectors in $\V^{k,j}$ of a given weight.

For generic $k,j$ the space of singular vectors in $\V^{k,j}_{[rl,(r+a)l]}$ is at most one-dimensional by Corollary~\ref{cor:reducibility_dim1}. Another uniqueness result for singular vectors was obtained in \cite{RCW}. It asserts that\footnote{The result concerns a certain class of Lie algebras. Here we state it for $\asl$ in more explicit terms.} if $k \in \mathbb N$ and $2j,2j'$ are integers not congruent to $-1$ modulo $k+2$, then the space of weight $j'$ singular vectors in $\V^{k,j}$ is at most one-dimensional. In \cite{Hwang} the same conclusion was obtained under the weaker assumptions $k+2 \in \Z \setminus \{ 0 \}$ and $2 j \in \Z$, using a similar argument. To the best of our knowledge, the question of uniqueness of singular vectors for general $k,2j$ is not discussed in the literature.

It was asked in \cite{KK} how to explicitly write down formulas for singular vectors. A~solution to this problem was obtained by Malikov, Feigin and Fuchs in \cite{MFF}. Below we present its equivalent reformulation in a language slightly different than that used in \cite{MFF}.

Malikov, Feigin and Fuchs have introduced the following operators
\begin{align}
    \mathcal O^+_{r,l}(t) &= (\sJ_0^-)^{l+rt} (\sJ_{-1}^+)^{l+(r-1)t} (\sJ_0^-)^{l+(r-2)t}  \cdots (\sJ_0^-)^{l-rt}, \label{eq:MFF_operators} \\
    \mathcal O^-_{r,l}(t) & = (\sJ_{-1}^+)^{l+(r-1)t} (\sJ_0^-)^{l+(r-2)t} \cdots (\sJ_{-1}^+)^{l-(r-1)t}, \nonumber
\end{align}
where $l \geq 1$ and $r \geq 0$ for $\mathcal O^+_{r,l}(t)$, $r \geq 1$ for $\mathcal O^-_{r,l}(t)$. This definition makes sense literally only if $t$ is an integer such that $r|t| \leq l$ (for $\mathcal O^+_{r,l}(t)$) or $(r-1)|t| \leq l$ (for $\mathcal O^-_{r,l}(t)$). However, as~we explain below, there is a natural interpretation of \eqref{eq:MFF_operators} as an element of $\U \aslm$ depending polynomially on $t \in \C$. Taking $t=k+2$ one obtains a family
\begin{equation}
    \chi^a_{r,l} = \mathcal O^a_{r,l} (k+2) v_{k,j} \in \V^{k,j}_{[rl,(r+a)l]}
    \label{eq:def_chi}
\end{equation}
depending polynomially on $k,j$.  If $T^a_{r,l}(k,j)=0$, $\chi^a_{r,l}$ is a singular vector, see Proposition \ref{prop:O_properties} and the discussion below.

Note the special cases
\begin{align}
    \mathcal O^+_{0,l}  = (\sJ_0^-)^l, \qquad  \qquad   \mathcal O^-_{1,l}  = (\sJ_{-1}^+)^l, \label{eq:elementary_O_operators}
\end{align}
in which the operators are independent of $t$. An elementary calculation shows that $\mathcal O^+_{0,l}v_{k,j}$ is a~singular vector of weight $-j-1$ if $2j+1=l$, and $\mathcal O^-_{1,l} v_{k,j}$ is a singular vector of weight $-j+k+1$ if $k-2j+1 =l$.

There exist canonical embeddings of $\U \aslm  \subset \U \asl$ as subalgebras in associative algebras $A^- \subset A$ in which $\sJ_0^-$ and $\sJ_{-1}^+$ are invertible elements, see Appendix \ref{app:localization}. If $t \in \Z $, then all exponents in \eqref{eq:MFF_operators} are integers, so \eqref{eq:MFF_operators} makes sense in $A^-$. Manipulating these expressions for integer $t$ allows to deduce existence of an extension to all $t$ which is a polynomial. We note that it would not be enough to consider only $t$ such that all exponents in $\eqref{eq:MFF_operators}$ are non-negative integers: since this is a finite set, we would not have uniqueness of continuation. In \cite{MFF} a different (in~our opinion unwieldy) approach to calculate with $\mathcal O^a_{r,l}(t) $ was adopted. It was also remarked that calculations could be performed in some enlarged algebra which contains complex powers of $J_0^-$ and $J_{-1}^+$, but precise definition of this algebra was not provided.

\begin{lemma} \label{lem:O_nonsingular}
If $t \in \Z$, then $\mathcal O^a_{r,l}(t)$ is in $\U \aslm$.
\end{lemma}
\begin{proof}
We may assume that $r |t| > l$, otherwise the statement is obvious. We assume also that $t >0$ and $a=+$. Proofs in the other three cases are analogous. We divide with remainder
\begin{equation}
    l = \alpha t + \beta, \qquad \alpha \in \{ 0 , \dots, r-1 \}, \, \beta \in \{ 0, \dots, t-1 \}.
\end{equation}
Then the product in \eqref{eq:MFF_operators} contains $r+1+\alpha$ factors with positive exponents. We assume that this number is odd (again for typographical reasons, the case of even $r+1+\alpha$ is handled analogously). The statement we want to prove is equivalent to existence of a solution (necessarily unique) $T \in \U \aslm$ to the equation
\begin{equation}
    T (\sJ_0^-)^{- \beta + (r- \alpha) t} \cdots (\sJ_{-1}^+)^{- \beta + t} = (\sJ_0^-)^{\beta + (\alpha + r)} \cdots (\sJ_0^-)^{\beta}.
    \label{eq:O_ratio}
\end{equation}
Now consider the vector $v_{t-2, \frac{\beta-1}{2}}$. It is easy to check that 
\begin{equation}
    x= (\sJ_0^-)^{- \beta + (r- \alpha) t} \cdots (\sJ_{-1}^+)^{- \beta + t} v_{t-2, \frac{\beta-1}{2}} 
\end{equation}
is a singular vector of weight $j = - \frac{1}{2}((r-\alpha)t-\beta+1)$ satisfying $T^+_{r,l}(t-2,j)=0$. By existence of singular vectors (applied to the module generated by $x$, which is isomorphic to $\V^{t-2,j}$), there exists $T \in \U \aslm$ such that $y=Tx$ is a singular vector of weight $j-l = - \frac12 ((r+\alpha)t+\beta+1)$. Module $\V^{t-2,\frac{\beta-1}{2}}$ contains also the singular vector 
\begin{equation}
    z = (\sJ_0^-)^{\beta + (\alpha+r)t} \cdots (\sJ_0^-)^{\beta} v_{t-2,\frac{\beta-1}{2}}
\end{equation}
of the same weight $j-l$. By uniqueness of singular vectors in $\V^{t-2,\frac{\beta-1}{2}}$, vectors $y$ and $z$ are proportional. Since elements of $\U \aslm$ acts as injective endomorphisms, this is (up to rescaling of~$T$) equivalent to the equation \eqref{eq:O_ratio}. 
\end{proof}

\begin{proposition}\label{prop:O_are_poly}
Fix $a,r,l$. The elements $\mathcal O^a_{r,l}(t) \in \U \aslm$ ($t \in \Z$) described in Lemma \ref{lem:O_nonsingular} coincide with the restriction to $\Z$ of a unique polynomial in $t$ with coefficients in $\U \aslm$. 
\end{proposition}
\begin{proof}
Uniqueness is clear because a polynomial is uniquely determined by its restriction to $\Z$. 
We prove existence by induction. Cases $(a,r)=(+,0)$ and $(a,r)=(-,1)$ are immediate from \eqref{eq:elementary_O_operators}. Now suppose that for some $r_0$ the statement is true for $(+,r_0-1)$ and $(-,r_0)$. We have
\begin{equation}
    \mathcal O_{r_0,l}^+(t) = (\sJ_0^-)^{l+rt} \mathcal O_{r_0,l}^-(t) (\sJ_0^-)^{l-rt}. 
\end{equation}
By the inductive hypothesis, there exists $N \in \N$ such that $[\sJ_0^-,\cdot]^{N+1}$ (the $(N+1)$-fold application of the commutator with $\sJ_0^-$) annihilates $\mathcal O_{r_0,l}^-(t)$, so
\begin{equation}
    \mathcal O_{r_0,l}^+(t) = \sum_{i=0}^N \binom{l+rt}{i} [\sJ_0^-,\cdot]^i (\mathcal O_{r_0,l}^-(t)) (\sJ_0^-)^{2l-i}.
    \label{eq:O_commutator_recursion}
\end{equation}
This expression is manifestly a polynomial in $t$ with coefficients in $A^-$. Since it is in $\U \aslm$ for all $t \in \Z$, it is a polynomial with coefficients in $\U \aslm$. $\mathcal O^-_{r_0+1,l}(t)$ can be treated analogously.
\end{proof}

We take polynomials provided by Proposition \ref{prop:O_are_poly} as the definition of $\mathcal O^a_{r,l}(t)$ for any $t \in \C$. We will now discuss basic properties of operators $\mathcal O^a_{r,l}(t)$ and singular vectors constructed with them. For that, it will be useful to have the operation of differentation with respect to generators on $\U \aslm$, described in the following Lemma.

\begin{lemma}
    There exist unique linear maps $\frac{\partial}{\partial \sJ^-_0}$, $\frac{\partial}{\partial \sJ^+_{-1}} : \U \aslm \to \U \aslm$ satisfying  Leibniz rule and:
    \begin{align}
         & \frac{\partial}{\partial \sJ^-_0}(\sJ_0^-) =1 , \qquad \frac{\partial}{\partial \sJ^-_0}(\sJ_{-1}^+) =0, \\
         & \frac{\partial}{\partial \sJ^+_{-1}}(\sJ_0^-) =0 , \qquad \frac{\partial}{\partial \sJ^+_{-1}}(\sJ_{-1}^+) =1. \nonumber
    \end{align}
$\frac{\partial}{\partial \sJ^-_0}$, $\frac{\partial}{\partial \sJ^+_{-1}}$ extend uniquely to derivations of $A^-$. We have $\frac{\partial}{\partial \sJ^-_0} \frac{\partial}{\partial \sJ^+_{-1}} = \frac{\partial}{\partial \sJ^+_{-1}} \frac{\partial}{\partial \sJ^-_0}$. 
\end{lemma}
\begin{proof}
The first part is an example of the following general construction: if $\mathfrak l$ is a Lie algebra and $\delta : \mathfrak l \to \C$ is a linear functional annihilating $[\mathfrak l, \mathfrak l] $, then $\delta$ extends to a derivation $\U \mathfrak l \to \U \mathfrak l$. Indeed, we first extend $\delta$ to the tensor algebra of $\mathfrak l$ by
\begin{equation}
    \delta \left( x_1 \otimes \cdots \otimes x_n \right) = \sum_{i=1}^n \delta(x_i) x_1 \otimes \cdots \otimes x_{i-1} \otimes x_{i+1} \otimes \cdots \otimes x_n.
\end{equation}
Then $\delta$ is a derivation and preserves the ideal generated by elements $x \otimes y - y \otimes x - [x,y]$, so it descends to a derivation on $\U \mathfrak l$.

Next we extend (say, $\frac{\partial}{\partial \sJ^-_0}$) to $A^-$. We have a homomorphism of algebras $\U \aslm \to A^- [\varepsilon]$ ($\varepsilon$ being an additional commuting generator satisfying $\varepsilon^2=0$) given by
\begin{equation}
    T \mapsto T + \varepsilon \frac{\partial}{\partial \sJ^-_0}(T).
    \label{eq:derivation_homo}
\end{equation}
Every monomial in $\sJ_0^-$ and $\sJ_{-1}^+$ is mapped to a sum of an invertible element and a nilpotent element, thus an invertible element. Hence the map \eqref{eq:derivation_homo} extends uniquely to an algebra homomorphism $A^- \to A^-[\varepsilon]$. It is easy to check that it is of the form \eqref{eq:derivation_homo} with some derivation $\frac{\partial}{\partial \sJ^-_0}$ extended from $\U \aslm$. 

For the final claim, note that $\left [ \frac{\partial}{\partial \sJ^-_0} , \frac{\partial}{\partial \sJ^+_{-1}} \right]$ is a derivation annihilating both $\sJ_0^-$ and $\sJ_{-1}^+$.
\end{proof}

We remark that the following analog of Poincar\'e Lemma holds: if $f,g \in \U \asl$ are such that $\frac{\partial f}{\partial \sJ_0^-} = \frac{\partial g}{\partial \sJ_{-1}^+} $, there exists $h$ such that $f = \frac{\partial h}{\partial \sJ_{-1}^+}$, $g = \frac{\partial h}{\partial \sJ_0^-}$. As we do not need this fact, we omit the proof.

Below we use the multi-index notation as in Subsection \ref{sec:KK_determinant}. For a multi-index $I=(i_{0,-},i_{1,+},\dots)$ we define the length
\begin{equation}
    |I| = \sum i_{l,a} = i_{0,-} + \sum_{l=1}^\infty \sum_a i_{l,a}.
\end{equation}
Note that if $I \in \mathcal P_{n,m}$, then $|n-m |\leq |I| \leq n+m$. The maximal value $|I|=n+m$ is attained only for the multi-index $I$ corresponding to $\sJ_{-I} = (\sJ_{-1}^+)^n (\sJ_0^-)^m$.

\begin{proposition} \label{prop:O_properties}
$\mathcal O^a_{r,l}(t)$ takes the form
\begin{equation}
\mathcal O^a_{r,l}(t) = \sum_{I \in \mathcal P_{rl,(r+a)l}} \mathcal O^a_{r,l,I}(t)\sJ_{-I},\label{eq:O_basis_expansion}
\end{equation}
where $\mathcal O^a_{r,l,I}(t)$ is a scalar polynomial of degree bounded by $(2r+a)l -|I|$. We have $\mathcal O^a_{r,l,I}(t)=1$ for the multi-index $I=((r+a)l,rl,0,0,\dots)$ of maximal $|I|$. In particular $\mathcal O^a_{r,l}(t) \neq 0$ for every $t \in \C$ and the degree of $\mathcal O^a_{r,l}(t)$ is bounded by $(2r+a-1)l$.

We have commutation rules
    \begin{align}
        [\sJ_0^+,\mathcal O^{+}_{r,l}(t)] &= \frac{\partial \mathcal O^+_{r,l}(t)}{\partial \sJ^-_0} (rt - l + 2\sJ_0^0+1), \nonumber \\
        [\sJ_1^-,\mathcal O^{+}_{r,l}(t)] &= -\frac{\partial \mathcal O^+_{r,l}(t)}{\partial \sJ^+_{-1}} (rt - l + 2\sJ_0^0+1 + (t-\mathsf{K}-2)), \label{eq:JO_commutators}  \\
        [\sJ_0^+,\mathcal O_{r,l}^-(t)] & = - \frac{\partial \mathcal O_{r,l}^-(t)}{\partial \sJ_0^-} (rt-l-2\sJ_0^0-1), \nonumber \\
        [\sJ_1^-, \mathcal O_{r,l}^-(t)] & = \frac{\partial \mathcal O_{r,l}^-(t)}{\partial \sJ_{-1}^+} (rt-l-2\sJ_0^0 -1 - (t-\mathsf{K}-2)). \nonumber
    \end{align}
\end{proposition}
\begin{proof}
It is sufficient to prove the claims for $t \in \Z$. To obtain \eqref{eq:O_basis_expansion}, we induct on $r$ as in the proof of Proposition \ref{prop:O_are_poly}. As a representative example, we consider $\mathcal O^+_{r,l}(t)$ for $l-rt \leq 0$. We have
\begin{equation}
    \mathcal O^+_{r,l}(t) (\sJ_0^-)^{rt-l} = (\sJ_0^-)^{l+rt} \mathcal O^-_{r,l}(t). \label{eq:Oplus_from_Ominus}
\end{equation}
The induction hypothesis for $\mathcal O^-_{r,l}(t)$ and the Poincar\'{e}-Birkhoff-Witt theorem imply that
\begin{equation}
    (\sJ_0^-)^{l+rt} \mathcal O^-_{r,l}(t) = \sum_{I \in \mathcal P_{rl,(r-1)l}} c^-_{r,l,I}(t) (\sJ_0^-)^{l+rt} \sJ_{-I} =\sum_{I \in \mathcal P_{rl,r(l+t)}} d_I(t) \sJ_{-I},
\end{equation}
and $d_I(t)=1$ for maximal $|I|$. By Poincar\'e-Birkoff-Witt and \eqref{eq:Oplus_from_Ominus}, we have $d_{I}(t) =0$ for multi-indices $I$ with $i_{0,-} \leq rt-l$ and $\deg(d_I(t)) \leq r(2l+t) - |I|$. Performing the division we get the claim for $\mathcal O^+_{r,l}(t)$.

Next we explain how to get the first line in \eqref{eq:JO_commutators}, the other being analogous. First check that for any $n \in \Z$ we have 
\begin{equation}
    [\sJ_0^+,(\sJ_0^-)^n] = n (\sJ_0^-)^{n-1} (2\sJ_0^0-n+1).
\end{equation}
Now we compute using the Leibniz rule
\begin{align}
   [\sJ_0^+, \mathcal O^+_{r,l}(t)] = (l+rt) (\sJ_0^-)^{l+rt-1} (2\sJ_0^0 - l-rt+1) (\sJ_{-1}^+)^{l+(r-1)t} \cdots (\sJ_0^{-})^{l-rt} + \dots, \label{eq:JO_commutator_step1}
\end{align}
where dots indicate terms where $[\sJ_0^+, \cdot]$ acts on subsequent factors in $\mathcal O^+_{r,l}(t)$ which are powers of $\sJ_0^-$. Now consider dragging the factor $(2\sJ_0^0 - l-rt+1)$ through all elements to its right. A simple calculation shows that it then becomes $rt-l+2\sJ_0^0+1$, and that exactly the same monomial is obtained from commuting the corresponding terms arising from the omitted terms in \eqref{eq:JO_commutator_step1}. 
\end{proof}

Recall the definition in \eqref{eq:def_chi}. By Proposition \ref{prop:O_properties}, $\chi^a_{r,l}$ is nonzero for every $k,j$ and we have
\begin{align}
    \sJ_0^+ \chi^a_{r,l} &= a \, T^a_{r,l}(k,j) \frac{\partial \mathcal O^a_{r,l}}{\partial \sJ_0^-} \chi^a_{r, l }, \\
    \sJ_{1}^- \chi^a_{r,l} & = - a \, T^a_{r,l}(k,j) \frac{\partial \mathcal O^a_{r,l}}{\partial \sJ_{-1}^+} \chi^a_{r,l}. \nonumber
\end{align}
In particular, if $T^a_{r,l}=0$, then $\chi^a_{r,l}$ is a singular vector. 

We remark that degree bounds in Proposition \ref{prop:O_properties} are sharp in the explicit examples of $\chi^a_{r,l}$, see \eqref{eq:singular_first_example} and below.

\subsection{Normalization factors of singular vectors} \label{sec:norm}

Thanks to the close relationship between $2-$dim CFT models with Virasoro and $\asl$ chiral symmetry algebras \cite{Ribault:2005wp,Hikida:2007tq}, many constructions and results from the former models have their counterparts in the latter ones. This concerns in particular norms of the so called Virasoro logarithmic primary fields 
in \cite{Zamolodchikov:2003yb,Yanagida:2010qm}. They are crucial ingredients of the Zamoldchikov's recurrence representation of the Virasoro conformal blocks
\cite{Zamolodchikov:ie,Zam0,Zam} and important elements of the AGT relation, where they appear as certain blow-up factors for the instanton moduli space 
\cite{Nekrasov:2002qd}. For the $\asl$ highest weight modules the corresponding normalization factor is defined in \eqref{eq:norms:definition}.

Polynomial $g(\chi^a_{r,l},\chi^a_{r,l})$ vanishes for $T^a_{r,l}=0$, so it is divisible by $T^a_{r,l}$. We set
\begin{equation}
    \label{eq:norms:definition}
    N^a_{r,l} = \left. \frac{g(\chi^a_{r,l},\chi^a_{r,l})}{T^a_{r,l}} \right|_{T^a_{r,l}=0}.
\end{equation}
We regard $N^a_{r,l}$ as a polynomial in $k$ (eliminating $j$ using $T^a_{r,l}=0$). The next Lemma will allow to conclude that $N^a_{r,l}$ is not identically zero.

\begin{lemma} \label{lem:inv_metric_lemma}
    Let $V$ be a finite-dimensional complex vector space and $B_{\epsilon}$ a family of bilinear forms on $V$ depending polynomially on $\epsilon$. Suppose that for some $t \in \N \setminus \{  0 \}$ we have
    \begin{itemize}
        \item relative to some (and hence any) basis of $V$, $\det (B_{\epsilon}) = c \epsilon^t + O(\epsilon^{t+1})$ for some $c \neq 0$,
        \item there exists a subspace $K \subset \Ker(B_0)$ of dimension $t$.
    \end{itemize}
Then $K = \Ker(B_0)$. $g(\cdot,\cdot) = \lim \limits_{\epsilon \to 0} \frac{1}{\epsilon} B_{\epsilon}(\cdot,\cdot)$ is a non-degenerate bilinear form on $K$. 
\end{lemma}
\begin{proof}
Let $l = \dim(\Ker (B_0))$. Choose a basis $e_1, \dots, e_l$ of $\Ker(B_0)$ and complete to a basis $e_1, \dots, e_N$ of $V$. Relative to this basis, the first $l$ rows and columns of $B_{\epsilon}$ are divisible by $\epsilon$, so $\det(B_{\epsilon})$ is divisible by $\epsilon^l$. Hence $l \leq t$ and thus $\Ker(B_0) = K$.

Now let $C$ be a complement of $K$ in $V$ and let $h_{\epsilon}$ be the restriction of $B_{\epsilon}$ to $C$ and $g_{\epsilon}$ the restriction of $\frac{1}{\epsilon} B_{\epsilon}$ to $K$. Clearly $h_0$, and hence $h_{\epsilon}$ for all sufficiently small~$\epsilon$, is nondegenerate. We have a block form
\begin{equation}
    B_{\epsilon} = \begin{bmatrix}
    h_{\epsilon} & \epsilon \phi_{\epsilon} \\
  \epsilon \phi^T_{\epsilon} & \epsilon g(\epsilon)
    \end{bmatrix}.
\end{equation}
Thus 
\begin{equation}
    \det(B_{\epsilon}) = \det(h_{\epsilon}) \det(\epsilon g_{\epsilon} - \epsilon^2 \phi^T_{\epsilon} h_{\epsilon}^{-1} \phi_{\epsilon}) = \epsilon^t \det(h_0) \det(g_0) + O(\epsilon^{t+1}).
\end{equation}
We conclude that $g=g_0$ has a nonzero determinant.
\end{proof}

\begin{proposition} \label{prop:N_zeros_constraint}
$N^a_{r,l}$ is nonzero if $k \neq -2$ and $T^{a'}_{r',l'} \neq 0$ for all $(a',r',l') \neq (a,r,l)$ such that $r'l' \leq rl$ and $(r'+a')l' \leq (r+a)l$. In particular, $N^a_{r,l}$ is not identically zero and all its roots are rational.
\end{proposition}
\begin{proof}
    Follows from Lemma \ref{lem:inv_metric_lemma} and \eqref{eq:Kac_det}.
\end{proof}

We have evaluated by explicit computations a few examples of $N^a_{r,l}$. The formulas appear most symmetric when expressed in terms of the variable $t := k+2$. 
\begin{itemize}
    \item $N_{0,l}^+=N_{1,l}^- =  l ((l-1)!)^2 $, 
    \item $N_{2,1}^+ = \left(t-1\right)t^2\left(t+1\right)
    \left(2t-1\right)\left(2t\right)^2\left(2t+1\right)$,
    \item $N_{1,2}^+ =  2\left(t-2\right)\left(t-1\right)^2 t^2 \left(t+1\right)^2\left(t+2\right)$. 
\end{itemize}
We remark that these examples show that constraints on zeros of $N^a_{r,l}$ in Proposition \ref{prop:N_zeros_constraint} are not sharp. For example Proposition \ref{prop:N_zeros_constraint} does not exclude $N_{2,1}^+$ vanishing at $k+2 = \frac{2}{3}$.

We conjecture that, up to a sign, $N_{r,l}^+$ and $N_{r+1,l}^-$ are given by
\begin{equation}
\label{eq:singular:vector:norms:result}
 \prod_{\substack{m \in \{ 1-l, \dots, l \} \\ q \in \{ -r, \dots , r \} \\ (m,q) \neq (0,0) }} (q(k+2)+m). 
\end{equation}
This conjecture is based on the explicit examples above and the form of the norm of Virasoro logarithmic primaries, see \cite{Yanagida:2010qm}, Theorem 1.2.

\subsection{Free fermion construction} \label{sec:fermi}

Let $\f$ be the Lie superalgebra spanned by $\{ \psi_n^i, \psibar_n^i \}_{n \in \Z + \frac12}^{i =1,2} \cup \{ I \}$, with all $\psi, \psibar$ odd, $I$ even and only nonzero superbrackets
\begin{equation}
[\psi_n^i, \psibar_m^j] = \delta_{n+m,0} \delta^{i,j} I.
\end{equation}
Fermionic Fock space $\F$ is defined as the quotient of $\U \f$ by the left ideal generated by $\{ \psi_n^i, \psibar_n^i \}_{n >0}^{i=1,2} \cup \{ I -1 \}$. We let $f_0$ be the image of $1$ in $\F$ and let 
\begin{equation}
f_{\frac12} = \psibar_{-\frac12}^1 f_0.
\end{equation}

$\F$ is acted upon by local fields
\begin{equation}
    \psi^i(z) = \sum_{n \in \Z + \frac12} \frac{\psi_n^i}{z^{n+ \frac12}}, \qquad \psibar^i(z) = \sum_{n \in \Z + \frac12} \frac{\psibar_n^i}{z^{n+ \frac12}}
\end{equation}
satisfying the OPE
\begin{equation}
    \psi^i(z) \psibar^j(w) \sim \frac{1}{z-w} \delta^{i,j}.
\end{equation}
We define currents (denoted $K^a$ rather than $\sJ^a$ to avoid a clash of notation later)
\begin{equation}
    K^+(z) = \psibar^1(z)\psi^2(z),
    \hskip 10mm
    K^-(z) = \psibar^2(z)\psi^1(z),
\end{equation}
and
\begin{equation}
    K^0(z)= \frac12\left(:\hskip -2pt\psibar^1(z)\psi^1(z)\hskip -3pt:-:\hskip -2pt\psibar^2(z)\psi^2(z)\hskip -3pt:\right).
\end{equation}
By Wick's formula, they satisfy the OPE \eqref{eq:current_OPE} with $k=1$, so $\F$ is an $\asl$-module at level $1$. Vectors $f_0, f_{\frac12}$ are both annihilated by $\aslp$ and $K_0^0 f_{\epsilon} = \epsilon f_{\epsilon}$ for $\epsilon \in \{ 0 , \frac12 \}$.

Let $\langle \cdot | \cdot \rangle$ be the symmetric bilinear form on $\F$ uniquely determined by conditions
\begin{equation}
    \langle f_0 | f_0 \rangle = 1 , \qquad \forall f, f' \in \F \ \langle f | \psi^i_n f' \rangle = \langle \psibar^i_{-n} f | f' \rangle.
\end{equation}
Then we have also
\begin{equation}
    \langle f | K^a_n f' \rangle = \langle K^{-a}_{-n} f | f' \rangle.
    \label{eq:fJ_conj}
\end{equation}

\begin{proposition} \label{prop:Fermi_Rep}
The $\asl$-submodule of $\F$ generated by $f_{\epsilon}$ is isomorphic to $\H^{1, \epsilon}$.
\end{proposition}
\begin{proof}
We have a homomorphism $h : \V^{1, \epsilon} \to \F$ taking $v_{1 , \epsilon}$ to $f_{\epsilon}$. Formula \eqref{eq:fJ_conj}, uniqueness of $g$ on $\V^{1 , \epsilon}$ and $\langle f_\epsilon | f_{\epsilon} \rangle =1$ imply that
\begin{equation}
\forall v,v' \in \V^{1, \epsilon} \    g(v,v') = \langle h(v) | h(v') \rangle.
\end{equation}
It is easy to check that $\langle \cdot | \cdot \rangle$ restricted to the image of $h$ has trivial kernel, so $\rad(\V^{1, \epsilon}) \subset \Ker(h)$ by Proposition \ref{prop:gker}. Thus $\Ker(h) = \rad(\V^{k,j})$ because $h \neq 0$.
\end{proof}

\begin{remark}
\label{remark:dual:fermion:currents}
For the purpose of later use, we have embedded $\H^{1, 0} \oplus \H^{1, \frac12}$ in $\F$. We emphasize that this is not a complete decomposition of $\F$ into irreducible representations of $\asl$. To find the latter, is is useful to consider additional fields
\begin{equation}
    {\overline K}^0(z)=\frac12\sum_i :\psibar^i(z) \psi^i(z):, \qquad {\overline K}^+(z)= \psibar^2(z) \psibar^1(z), \qquad {\overline K}^-(z) =\psi^1(z) \psi^2(z),
\end{equation}
whose coefficients span a second $\asl$ algebra commuting with all $K^a_n$.
\end{remark}


\section{Wakimoto representation} \label{sec:Wakimoto}

\subsection{Fock modules}

Fix $\kappa \neq 0$ and define $k$ by $\kappa^2 = k+2$. We consider the Lie algebra $\w$ spanned by elements $a_n, \beta_n, \gamma_n, I$, $n \in \mathbb Z$, with commutators:
\begin{equation}
    [\gamma_m, \beta_n] = \delta_{m+n} I, \qquad [ a_m,  a_n] = \frac{m}{2} \delta_{m+n} I.
\end{equation}

We define a $\w$-module $\W^{\kappa,j}$, called (Wakimoto) Fock module, as the quotient of $\U \w$ by the left ideal generated by $\{ a_0 - \kappa^{-1} j, I-1 \} \cup \{ \beta_n, \gamma_{n+1}, a_{n+1} \}_{n = 0}^{\infty}$. $\W^{\kappa,j}$ is an irreducible representation of $\w$. Let $w_{\kappa,j}$ be the image of $1$ in $\W^{\kappa,j}$.

We introduce the following local fields on $\W^{\kappa,j}$:
\begin{equation}
\label{free:fields:definition}
\partial \phi(z) = \sum_{n \in \Z} \frac{a_n}{z^{n+1}}, \qquad \beta(z) = \sum_{n \in \Z} \frac{\beta_n}{z^{n+1}}, \qquad \gamma(z) = \sum_{n \in \Z} \frac{\gamma_n}{z^n}.
\end{equation}
Here are their OPEs:
\begin{equation}
  \gamma(z) \beta(w) \sim \frac{1}{z-w}, \qquad \partial \phi(z) \partial \phi(w) \sim \frac12 \frac{1}{(z-w)^2}.
  \label{eq:free_OPE}
\end{equation}
Next, we define currents:
\begin{align}
    J^+(z) & = \beta(z), \nonumber \\
    J^0(z) & = :\gamma(z) \beta(z): + \kappa \partial \phi(z), \label{eq:free_currents} \\
    J^-(z) & = - :\gamma(z)^2 \beta(z): - 2 \kappa \gamma(z) \partial \phi(z) - k \partial \gamma(z). \nonumber
\end{align}
\begin{remark}
To avoid a possible confusion, we use the sans serif font $\sJ^a(z)$ to denote the ``abstract'' $\asl$ level $k$ currents, and italic font $J^a(z),\widetilde{J}^a(z),$ to denote free field realizations of these currents in, respectively, the Fock module and the dual Fock module (see Subsection \ref{ssect:dual:Fock:module}). On the other hand, since we embedded 
$\H^{1,\epsilon}$ in the free fermion space $\F$ we shall denote both the abstract $\asl$ level one currents and their realization in $\F$ using  the italic font $K^a(z).$
\end{remark}
Using \eqref{eq:free_OPE} and Wick's formula for OPE one can check that \eqref{eq:current_OPE} holds. Hence $\W^{\kappa,j}$ is an $\asl$-module, with operators $J^a_n$ defined as coefficients of the expansion \eqref{eq:current_modes}. One checks that $J_0^0 w_{\kappa,j} = j w_{\kappa,j}$ and $\aslp w_{\kappa,j} =0$. We note the commutation rules:
\begin{equation}
    [J^0_0,\beta_n] = \beta_n, \qquad [J_0^0,\gamma_n] = - \gamma_n, \qquad [J^0_0,a_n] =0.
    \label{eq:J00_com}
\end{equation}

\begin{proposition}
Sugawara's construction with fields \eqref{eq:free_currents} gives
\begin{equation}
    T(z) = - :\beta(z) \partial \gamma(z) : + :\partial \phi(z) \partial \phi(z): - \kappa^{-1} \partial^2 \phi(z).
    \label{eq:free_T}
\end{equation}
\end{proposition}
\begin{proof}
Checked with Wick's formula.
\end{proof}

The following commutation rules follow immediately from \eqref{eq:free_T}:
\begin{equation}
[L_0,\beta_n] = -n \beta_n, \qquad [L_0,\gamma_n] = -n \gamma_n, \qquad [L_0,a_n] =-n a_n.
\label{eq:L0_com}
\end{equation}

\begin{proposition} \label{prop:dim_comp}
$L_0$ and $J_0^0$ are diagonalizable in $\W^{\kappa,j}$ and $\mathrm{dim}(\mathcal W^{\kappa,j}_{\Delta',j'}) = \mathrm{dim}(\mathcal V^{k,j}_{\Delta',j'})$ for every $(\Delta',j')$.
\end{proposition}
\begin{proof}
Follows from Poincar\'{e}-Birkhoff-Witt theorem combined the eigenvalue equation satisfied by $w_{\kappa,j}$ and commutation rules (\ref{eq:J00_com}, \ref{eq:L0_com}). See also Subsection \ref{sec:bases} for a discussion of explicit bases indexed by the same sets.
\end{proof}

Let $s : \V^{k,j} \to \W^{\kappa,j}$ be the $\asl$-module map determined by $s (v_{k,j})= w_{\kappa,j}$.

The presence of a third order pole in the OPE
\[
T(z)\partial\phi(w) \sim \frac{\kappa^{-1}}{(z-w)^3} + \frac{\partial}{\partial w} \frac{\partial \phi(w)}{z-w}
\]
shows that $\partial\phi$ is not a primary field and transforms under conformal transformations generated by $T(z)$ (in particular under the inversion $z \to z^{-1}$) in an anomalous way. A manifestation of this anomalous behavior is that looking for a definition of the adjoint of $a_n$ leaving the form of the stress tensor (\ref{eq:free_T}) invariant, or computing the adjoint of the commutator
\[
[L_m,a_n]  = -na_{m+n} +\frac{\kappa^{-1}}{2}m(m+1)\delta_{m,-n}, 
\]
we are led to postulate that $a_n^* = -\kappa^{-1}\delta_{n,0}-a_n$. Such conjugation rule is not compatible with a scalar product on $\W^{\kappa,j}$ because $a_0$ acts on $\W^{\kappa,j}$ as a scalar. It is satisfied for a pairing with the dual module, which we will now describe.

\subsection{Dual Fock modules}
\label{ssect:dual:Fock:module}

$\w$-module $\Wt^{\kappa,j}$ is defined as the quotient of $\U \w$ by the left ideal generated by 
\begin{equation}
 \{ a_0 + \kappa^{-1} (j+1), I-1 \} \cup \{ \beta_{n+1}, \gamma_n, a_{n+1} \}_{n= 0}^\infty.   
\end{equation}
Let $\widetilde w_{\kappa,j}$ be the image of $1$ in $\Wt^{\kappa,j}$.

\begin{proposition}
\label{prop:pairing:W:tildeW}
There exists a unique bilinear form $\langle \cdot | \cdot \rangle  : \Wt^{\kappa,j} \times \W^{\kappa,j} \to \C$ such that $\langle \widetilde w_{\kappa,j}, w_{\kappa,j} \rangle = 1$ and for every $\varphi \in \Wt^{\kappa,j}$ and $w \in \W^{\kappa,j}$ one has:
\begin{gather}
    \langle \varphi | \beta_n w \rangle =  - \langle \beta_{-n} \varphi | w \rangle, \qquad \langle \varphi | \gamma_n w \rangle = \langle \gamma_{-n} \varphi | w \rangle,  \\
    \langle \varphi | a_n w \rangle = \langle (- \kappa^{-1} \delta_{n,0} - a_{-n}) \varphi | w \rangle. \nonumber
\end{gather}
$\langle \cdot | \cdot \rangle$ is non-degenerate.
\end{proposition}

We won't define currents acting in $\Wt^{\kappa,j}$ as in \eqref{eq:free_currents}, because then $\widetilde w_{\kappa,j}$ would not annihilated by $J_0^+$. Instead, we modify them by applying the automorphism $J_n^a \to - J^{-a}_n$ (corresponding to $\begin{bmatrix}
0 & 1 \\ -1 & 0
\end{bmatrix}$ in $\mathrm{SL}(2,\C)$ generated by elements $J_0^a$):
\begin{align}
    \widetilde{J}^+(z) & =  :\gamma(z)^2 \beta(z): + 2 \kappa \gamma(z) \partial \phi(z) + k \partial \gamma(z), \nonumber \\
        \widetilde{J}^0(z) & = -:\gamma(z) \beta(z): - \kappa \partial \phi(z), \label{eq:free_currents_dual} \\
    \widetilde{J}^-(z) & = -\beta(z). \nonumber
\end{align}
This redefinition does not affect the formula \eqref{eq:free_T} for the Sugawara field. Moreover,
\begin{equation}
    \langle \varphi | J^a_n v \rangle = \langle  \widetilde{J}_{-n}^{-a} \varphi | v \rangle.
\end{equation}
Pairing $\langle \cdot | \cdot \rangle$ exhibits $\Wt^{\kappa,j}$ as the contragradient dual of $\W^{\kappa,j}$. In particular, Proposition \ref{prop:dim_comp} remains true if $\W^{\kappa,j}$ is replaced by $\widetilde \W^{\kappa,j}$. A homomorphism $\widetilde s : \mathcal V^{k,j} \to \Wt^{\kappa,j}$ is uniquely determined by $\widetilde s(v_{k,j}) = \widetilde w_{\kappa,j}$. 

\begin{proposition} \label{prop:metrics_comparison}
Bilinear form $g$ on $\mathcal V^{k,j}$ may be written in the form
\begin{equation}
g(v,v') = \langle \widetilde s(v) | s(v') \rangle.
\label{eq:g_factorization}
\end{equation}
\end{proposition}
\begin{proof}
The right hand side defines a bilinear form which satisfies conditions known to uniquely determine $g$.
\end{proof}

\subsection{Determinant calculations} \label{sec:bases}

We fix a basis of $\V^{k,j}$ as in Section \ref{sec:KK_determinant}. In $\W^{\kappa,j}$ we choose $\{ \Gamma_{-I} w_{\kappa,j} \}_I$, where
\begin{equation}
\Gamma_{-I} = \left[ \prod_{l=1}^{\infty} a_{-l}^{i_{l,0}} \beta_{-l}^{i_{l,+}} \gamma_{-l}^{i_{l,-}} \right] \gamma_0^{i_{0,-}}.
\end{equation}
In the module $ \Wt^{\kappa,j}$ we will use basis $\{ \widetilde \Gamma_{-I} \widetilde w_{\kappa,j} \}_{I}$, where
\begin{equation}
\widetilde \Gamma_{-I} = \left[ \prod_{l=1}^{\infty} (-a_{-l})^{i_{l,0}} \gamma_{-l}^{i_{l,+}} (-\beta_{-l})^{i_{l,-}} \right] (-\beta_0)^{i_{0,-}}.
\end{equation}
Polynomial dependence of vectors, operators etc.\ on $\W^{\kappa,j}$ or $\Wt^{\kappa,j}$ always refers to expressions in these bases.

Matrices $S_{n,m}, \widetilde S_{n,m}$ are defined by
\begin{align}
s(\sJ_{-I} v_{k,j}) &= \sum_{J \in \mathcal P_{n,m}} \Gamma_{-J} w_{\kappa,j} S_{n,m}^{J,I},  \\
\widetilde s(\sJ_{-I} v_{k,j}) &= \sum_{J \in \mathcal P_{n,m}}  \widetilde \Gamma_{-J} \widetilde w_{\kappa,j} \widetilde S_{n,m}^{J,I}, \nonumber
\end{align}
for $I \in \mathcal P_{n,m}$. Their entries are polynomials in $\kappa, j$. On the other hand, matrix elements:
\begin{equation}
    \langle \widetilde \Gamma_{-I} \widetilde w_{\kappa,j} | \Gamma_{-J} w_{\kappa,j} \rangle, \qquad I,J \in \mathcal P_{n,m},
\end{equation}
do not depend on $\kappa,j$. Hence \eqref{eq:g_factorization} implies that there exist constants $c_{n,m}$ such that
\begin{equation}
    \det(G_{n,m}) = c_{n,m} \det(\widetilde S_{n,m}) \det(S_{n,m}).
    \label{eq:Gdet_factorized}
\end{equation}

\begin{lemma} \label{lem:det_asymptotic}
Consider grading on the algebra of complex polynomials in $\kappa$ and $j$ given by $\mathrm{deg}(\kappa)= 1$, $\mathrm{deg}(j)=2$. For every $I \in \mathcal P_{n,m}$, the degree of $S_{n,m}^{I,J}$ and $\widetilde S_{n,m}^{I,J}$ is~maximized for $J=I$. Terms of maximal degree take the form:
\begin{align}
S_{n,m}^{I,I} & \approx \kappa^{\sum \limits_{l=1}^{\infty} i_{l,0}} \prod_{l=0}^{\infty} (-2j-l \kappa^2 )^{i_{l,-}}, \label{eq:asymptotic_S_det} \\
\widetilde S_{n,m}^{I,I} & \approx \kappa^{\sum \limits_{l=1}^{\infty} i_{l,0}} \prod_{l=1}^{\infty} (-2j+l \kappa^2 )^{i_{l,+}}.
\end{align}
\end{lemma}
\begin{proof}
By inspection of the form of $J_n^a$ and $\widetilde{J}_n^a$.
\end{proof}

\begin{proposition} \label{prop:S_det}
\begin{align}
\det(S_{n,m}) &= \kappa^{\sum \limits_{r,l=1}^{\infty} d_{n-rl,m-rl}} \prod_{(r,l) \in \mathcal S^+} \left( -T^+_{r,l} \right)^{d_{n-rl,m-(r+1)l}}, \label{eq:S_det} \\
\det(\widetilde S_{n,m}) &= \kappa^{\sum \limits_{r,l=1}^{\infty} d_{n-rl,m-rl}} \prod_{(r,l) \in \mathcal S^-} \left( T^-_{r,l} \right)^{d_{n-rl,m-(r-1)l}}. \nonumber
\end{align}
\end{proposition}
\begin{proof}
Formula \eqref{eq:Gdet_factorized} implies that every factor of $\det(S_{n,m}), \det(\widetilde S_{n,m})$ is present in \eqref{eq:Kac_det}. Conversely, every irreducible factor of the latter divides either $\det(S_{n,m})$ or $\det(\widetilde S_{n,m})$ (with suitable multiplicities). Lemma \ref{lem:det_asymptotic} allows to decide which factors of \eqref{eq:Kac_det} can be found in which of $\det(S_{n,m}), \det(\widetilde S_{n,m})$.
\end{proof}

\begin{lemma} \label{lem:inverse_singularity}
Let $S_x$ be an $N \times N$ matrix depending polynomially on $x$. Suppose that the determinant of $S_{x}$ vanishes at $0$ to order $\dim \Ker(S_{x_0})$ (and not higher). Then the limits:
\begin{equation}
    A = \lim_{x \to 0} \left. \frac{S_x}{x} \right|_{\Ker(S_0)}, \qquad B = \lim_{x \to 0} x S_x^{-1},
\end{equation}
exist. We have that $\Ran(A)$ is a complement to $\Ran(S_0)$, $\Ker(B)=\Ran(S_0)$, and that $\Ran(B)= \Ker(S_0)$. Moreover, $BA= 1_{\Ker(S_0)}$, and $AB$ is a projection with image $\Ran(A)$ and kernel $\Ran(S_0)$.
\end{lemma}
\begin{proof}
Established by considering the Smith's normal form of $S_x$.
\end{proof}

\begin{proposition}
    We have:
    \begin{equation}
        \left. s(\chi^+_{r,l}) \right|_{T^+_{r,l}=0}=0, \qquad \left. \widetilde s(\chi^-_{r,l}) \right|_{T^-_{r,l}=0}=0.
    \end{equation}
The polynomial families
\begin{equation}
    \left. s(\chi^-_{r,l}) \right|_{T^-_{r,l}=0}, \qquad \left. \widetilde s(\chi^+_{r,l}) \right|_{T^+_{r,l}=0},
    \label{eq:WHW}
\end{equation}
and
\begin{equation}
    \left. \frac{s(\chi^+_{r,l})}{T^+_{r,l}} \right|_{T^+_{r,l}=0}, \qquad \left. \frac{\widetilde s(\chi^-_{r,l})}{T^-_{r,l}} \right|_{T^-_{r,l}=0}
    \label{eq:WNR}
\end{equation}
are not identically zero. Vectors in \eqref{eq:WHW} are highest weight vectors.
\end{proposition}
\begin{proof}
We consider the family $\chi^+_{r,l}$, the other case being analogous. $\widetilde s$ is generically injective for $T^+_{r,l}=0$, so $\left. \widetilde s(\chi^+_{r,l}) \right|_{T^+_{r,l}=0}$ is not identically zero. For generic $k$ and $T^+_{r,l}=0$ we have that $\chi^+_{r,l}$ generates the unique nonzero proper submodule of $\V^{k,j}$. By Proposition \ref{prop:S_det}, $\Ker(s)$ is then also a nonzero proper submodule and hence is generated by $\chi^+_{r,l}$. Therefore, $s(\chi^+_{r,l})$ is divisible by $T^+_{r,l}$. It is not divisible by $(T^+_{r,l})^2$ because then $\det(S_{rl,(r+1)l})$ would be divisible by $(T_{r,l}^+)^2$ (see~also Lemma \ref{lem:inverse_singularity}). 
\end{proof}

To illustrate the statement above, we provide a few simple examples.
\begin{itemize}
    \item We have $\chi^+_{0,l} =  (\sJ_0^-)^l v_{k,j}$, $T^+_{0,l}= -l +2j+1$ and
    \begin{equation}
    \label{eq:singular_first_example}
        \left. \frac{s(\chi^+_{0,l})}{T_{0,l}^+} \right|_{T_{0,l}^+=0} = (-1)^l (l-1)! \gamma_0^l w_{\kappa, \frac{l-1}{2}}, \quad
    \left. \widetilde s(\chi^+_{0,l}) \right|_{T_{0,l}^+=0} = (- \beta_0)^{l} \widetilde w_{\kappa, \frac{l-1}{2}}.
    \end{equation}
    \item $ \chi^-_{1,l}  = (\sJ_{-1}^+)^l v_{k,j}$, $T_{1,l}^- = k+2 -l - 2j -1$ and
    \begin{equation}
        \left. s(\chi_{1,l}^-) \right|_{T_{1,l}^-=0} = \beta_{-1}^l w_{\kappa, \frac{k+1-l}{2}}, \qquad \left. \frac{\widetilde s(\chi_{1,l}^-)}{T_{1,l}^-} \right|_{T_{1,l}^- =0} = (l-1)! \gamma_{-1}^l \widetilde w_{\kappa, \frac{k+1-l}{2}}.
    \end{equation}
    \item $\chi^+_{1,1}  =  \left(\sJ^+_{-1}(\sJ^-_0)^2 -2(k+3)\sJ^0_{-1}\sJ^-_0 - (k+2)(k+3)\sJ^-_{-1}\right)v_{k,j} $, $T_{1,1}^+ = k+2j+2$,
    \begin{align}
         \left. \frac{s(\chi^+_{1,1})}{T_{1,1}^+} \right|_{_{T_{1,1}^+=0}} &= 
         \left(k(k+3)\gamma_{-1} + 2(k+3)\kappa a_{-1}\gamma_0+ \beta_{-1}\gamma_0^2\right)w_{\kappa,-\frac{k+2}{2}},\\   
         \left. \widetilde{s}(\chi^+_{1,1})\right|_{T_{1,1}^+=0} &= (k+1)\kappa\left(\kappa\beta_{-1} -2 a_{-1}\beta_0\right)\widetilde w_{\kappa, -\frac{k+2}{2}}. \nonumber
    \end{align}
    \item $\chi^-_{2,1}  =  \left((\sJ^+_{-1})^2\sJ^-_0 + 2(k+1)\sJ^+_{-1}\sJ^0_{-1} -(k+1)(k+2)\sJ^+_{-2}\right) v_{k,j}$, $T_{2,1}^- = 2k-2j+2$,
    \begin{align}
       \left. s(\chi^-_{2,1}) \right|_{T_{2,1}^-=0} & =(k+1) \kappa(2 a_{-1}\beta_{-1}-\kappa\beta_{-2})w_{\kappa,k+1}, \\
       \left. \frac{\widetilde{s}(\chi^-_{2,1})}{T_{2,1}^-} \right|_{T_{2,1}^-=0} 
       &= -\left(2(k+3)\kappa a_{-1}\gamma_{-1} + k(k+3)\gamma_{-2} + \gamma_{-1}^2\beta_0\right)\widetilde w_{\kappa,k+1}. \nonumber
    \end{align}
\end{itemize}

The remaining part of this Subsection is devoted to understanding the singularities of $s^{-1}, \widetilde s^{-1}$.

\begin{proposition} \label{prop:sinv_sing_characterization}
    Suppose that $k \neq -2$ and all $T_{r,l}^-$ are nonzero. Then 
    \begin{equation}
        \Ker(s) = \rad(\V^{k,j}), \qquad \Ran(s) = \widetilde s(\rad(\V^{k,j}))^\perp.
    \end{equation}
     Similarly, if $k \neq -2$ and all $T_{r,l}^+$ are nonzero, then
     \begin{equation}
         \Ker(\widetilde s) = \rad(\V^{k,j}), \qquad \Ran( \widetilde s) = s(\rad(\V^{k,j}))^\perp.
     \end{equation} 
\end{proposition}
\begin{proof}
We prove the statement about the range of $s$, the other being analogous. Since $\widetilde s$ is an isomorphism and the pairing between $\widetilde \W^{\kappa,j}$ and $\W^{\kappa,j}$ is non-degenerate, $\Ker(s) = \rad(\V^{k,j})$. It is easy to check that the range of $s$ is contained in the orthogonal complement of $\widetilde s(\rad(\V^{k,j}))$. By a dimension count, these two modules are equal. 
\end{proof}

\begin{lemma} \label{lem:overlap_criterion}
Suppose that $u_{\kappa,j}$ is a family of vectors in $\W^{\kappa,j}$ depending polynomially on $\kappa,j$. Let $I$ be the set of pairs $(r,l)$ such that $\langle \widetilde s(T \chi^+_{r,l}), u_{\kappa,j} \rangle$ is not divisible by $T^+_{r,l}$ for every $T \in U \asl$. There exists $n \in \N$ such that
    \begin{equation}
        u'_{\kappa,j} := \kappa^n \prod_{(r,l) \in I} T_{r,l}^+ \cdot  s^{-1}(u_{\kappa,j})
        \label{eq:sinv_singularities}
    \end{equation}
    depends polynomially on $\kappa,j$. 

    An analogous statement is true for polynomial families in $\widetilde \W^{\kappa,j}$, with the roles of $s,\widetilde s$ interchanged and $T^+_{r,l}$, $\chi^+_{r,l}$ replaced by $T^-_{r,l}$, $\chi^-_{r,l}$.
\end{lemma}
\begin{proof}
    Coefficients of $s^{-1}(u_{\kappa,j})$ relative to the canonical basis are rational functions, with denominators containing only powers of $T^+_{r,l}$ and $\kappa$. Restricting to fixed generic $\kappa$ and using Lemma \ref{lem:inverse_singularity} (for $s$ restricted to some $\V^{k,j}_{[n,m]}$) we see that (if ratios are written as irreducible fractions) $T^+_{r,l}$ may occur only linearly in the denominator of $s^{-1}(u_{\kappa,j})$. By Proposition \ref{prop:sinv_sing_characterization}, it~occurs in the denominator of some coefficient if and only if $(r,l) \in I$.
\end{proof}

Next we determine $n$ in \eqref{eq:sinv_singularities}. For this it is sufficient to consider basis vectors.

\begin{lemma} \label{lem:kappa_singularity}
 For every multi-index $I$, vectors
    \begin{equation}
        \kappa^{\sum_{l=1}^\infty i_{l,0}} s^{-1} (\Gamma_{-I} w_{\kappa,j}), \qquad \kappa^{\sum_{l=1}^\infty i_{l,0}} \widetilde s^{-1} (\widetilde \Gamma_{-I} \widetilde w_{\kappa,j})
    \end{equation}
are nonsingular at $\kappa=0$.
\end{lemma}
\begin{proof}
We change basis from vectors $\Gamma_{-I} w_{\kappa,j}$ to $ \kappa^{\sum_{l=1}^\infty i_{l,0}} \Gamma_{-I} w_{\kappa,j}$. Then the determinant of $s$ in the new basis is not divisible by $\kappa$, cf.\ \eqref{eq:asymptotic_S_det} and \eqref{eq:S_det}. The same argument applies also for the dual module.
\end{proof}

\subsection{Screening charges and singular vectors} \label{sec:screenieng_charges}

There exist operators $\mathrm{e}^{2 t q_{\phi}} : \W^{\kappa,j} \to \W^{\kappa,j+ \kappa t}, \, \Wt^{\kappa,j} \to \Wt^{\kappa,j-\kappa t}$ with the following properties:
\begin{itemize}
    \item $\mathrm{e}^{2 t q_\phi} \mathrm{e}^{2t' q_{\phi}} = \mathrm{e}^{2(t+t')q_{\phi}}$ (justifying the exponential notation),
    \item $\mathrm{e}^{-2 t q_\phi} a_0 \mathrm{e}^{2 t q_{\phi}} = a_0 +t$ and $\mathrm{e}^{2t q_{\phi}}$ commutes with $\gamma_n,\beta_n$ and with $a_n$, $n \neq 0$ (exponentiated canonical commutation rules),
    \item $\mathrm{e}^{2t q_{\phi}} w_{\kappa,j} = w_{\kappa, j+\kappa t}$, $\mathrm{e}^{2 t q_{\phi}} \widetilde w_{\kappa,j}=\widetilde w_{\kappa,j-\kappa t}$.
\end{itemize}
It is clear that these rules determine $\mathrm{e}^{2 t q_{\phi}}$ uniquely.
\begin{proof}[Proof of existence]
The maps:
\begin{equation}
 \gamma_m \mapsto \gamma_{m}, \ \beta_m \mapsto \beta_{m}, \ a_m \mapsto a_m + \delta_{m,0} t, \ I \mapsto I, \qquad t \in \C
\end{equation}
are automorphisms of $\w$ taking the annihilator of $w_{\kappa ,j}$ in $\U \w$ to the annihilator $\w^{\kappa,j-\kappa t}$, and similarly for $\widetilde w_{\kappa,j}$.
\end{proof}

 

We can now introduce the (multi-valued) normal ordered exponential fields
\begin{eqnarray}
\label{eq:ordered:exponential:1}
{\mathsf E}^\alpha(z) &  \equiv &  :\!{\mathrm e}^{2\alpha\phi(z)}\!: \;\; = \;   {\mathrm e}^{2\alpha\sum\limits_{m=1}^\infty \frac{a_{-m}}{m}z^m}{\mathrm e}^{2\alpha q_\phi} z^{2\alpha a_0}
{\mathrm e}^{-2\alpha\sum\limits_{m=1}^\infty \frac{a_{m}}{m}z^{-m}}.
\end{eqnarray}
They satisfy\footnote{$(z_1-z_2)^{2 \alpha_1 \alpha_2}$ is understood as $z_1^{2 \alpha_1 \alpha_2} (1-z_1^{-1} z_2)^{2\alpha_1\alpha_2}$ with $(1-z_1^{-1} z_2)^{2\alpha_1\alpha_2}$ given by the series expansion convergent for $|z_1| > |z_2|$.}
\begin{equation}
{\mathsf E}^{\alpha_1}(z_1){\mathsf E}^{\alpha_2}(z_2)
=
(z_1-z_2)^{2\alpha_1\alpha_2}:\!{\mathsf E}^{\alpha_1}(z_1){\mathsf E}^{\alpha_2}(z_2)\!:
\end{equation}
and the OPE
\begin{equation}
\partial\phi(w){\mathsf E}^\alpha(z) \sim \frac{\alpha}{w-z}{\mathsf E}^\alpha(z).
\end{equation}

Consider now the field 
\(
\beta(w){\mathsf E}^{-\kappa^{-1}}(w)
\)
and (for $-\kappa^2\in \N \setminus \{ 0 \}$) $\beta^{-\kappa^2}(w){\mathsf E}^{\kappa}(w)$. We have:
\begin{equation}
\label{eq:OPE:Jplus:zero:screening:charges}
J^{\sharp}(z) \beta(w){\mathsf E}^{-\kappa^{-1}}(w) \sim \mathrm{reg},
\hskip 1cm
J^{\sharp}(z) \beta^{-\kappa^2}(w){\mathsf E}^{\kappa}(w) \sim \mathrm{reg},
\hskip 1cm \sharp =+,0,
\end{equation}
and
\begin{eqnarray}
\label{eq:OPE:Jminus:screening:charges}
\nonumber
J^-(z)\beta(w){\mathsf E}^{-\kappa^{-1}}(w) & \sim & \kappa^2\partial_w\left(\frac{1}{z-w}{\mathsf E}^{-\kappa^{-1}}(w)\right),
\\[-7pt]
\\[-7pt]
\nonumber
J^-(z)\beta^{-\kappa^2}(w){\mathsf E}^{\kappa}(w) & \sim & \kappa^2\partial_w\left(\frac{1}{z-w} \beta^{-\kappa^2-1}(w){\mathsf E}^{\kappa}(w)\right).
\end{eqnarray}
The \emph{screening charges} are defined by: 
\begin{equation}
\label{eq:screening_charges}
{\mathsf Q}_- = \int \beta(w){\mathsf E}^{-\kappa^{-1}}(w)\, dw, \hskip 1cm {\mathsf Q}_+ = \int \beta^{-\kappa^2}(w){\mathsf E}^{\kappa}(w)\,dw,
\end{equation}
with integrals performed over suitable contours. It is automatic that ${\mathsf Q}_-, {\mathsf Q}_+$ commute with all the $\asl$ generators $J^{\sharp}_n.$ Since right hand sides in \eqref{eq:OPE:Jminus:screening:charges} are full derivatives, equalities
\begin{equation*}
    [J^-_n,{\mathsf Q}_-] = [J^-_n,{\mathsf Q}_+] = 0
\end{equation*}
hold if there are no contribution from the endpoints of the integration contour (e.g.~the contour is closed). We will use symbols $\mathsf Q_\pm^l$ for $l-$fold iterated integrals, e.g.
\begin{equation}
    \mathsf Q_-^2 = \int \left( \int \beta(w_1) \mathsf E^{- \kappa^{-1}}(w_1) \beta(w_2) \mathsf E^{- \kappa^{-1}}(w_2) dw_2 \right) dw_1 .
\end{equation}
Since the contour of integration for $w_2$ typically depends on $w_1$ it is not accurate to think of $\mathsf Q_\pm^l$ as the $l-$th power of $\mathsf Q_\pm$.

\begin{proposition}
\label{prop:singular:vectors:through:screening:charges}
Let $j^a_{r,l}$ be the unique solution of $T^a_{r,l}(k,j^a_{r,l})=0$. Define
\begin{equation}
{\widetilde h}_{\kappa}^{(r,l)}= {\mathsf Q}_-^l \widetilde w_{\kappa,j^+_{r,l}-l} \in \Wt^{\kappa,j^+_{r,l}}, \qquad 
h_{\kappa}^{(r,l)}= {\mathsf Q}_-^lw_{\kappa,j^-_{r,l}+l} \in \W^{\kappa,j^-_{r,l}}.
\end{equation}
Integration contours in $\mathsf Q^l_-$ may be chosen so that these expressions are not identically zero, and are proportional to  $\widetilde s(\chi^+_{r,l})$ and $s(\chi^-_{r,l}),$ respectively.
\end{proposition}
\begin{proof}
Proof of Proposition \ref{prop:singular:vectors:through:screening:charges} can be found in \cite{Bernard:1989iy}, Section 2. 
\end{proof}
The integration contours chosen in \cite{Bernard:1989iy} form a set of nested
curves going counterclockwise from $1$ to $\mathrm{e}^{2\pi i}1,$ and are localized in the
neighborhood of the circle of radius 1 centered at the origin (as in \cite{Felder:1989}).

\vskip 5mm
\centerline{\includegraphics[width=7cm]{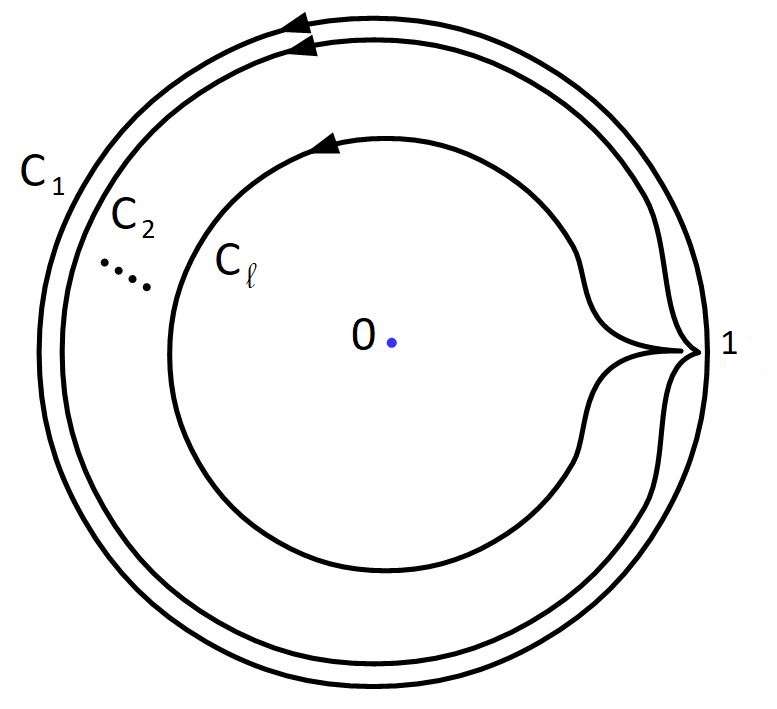}}

For $|z_1|> |z_2| > \ldots > |z_l|$ standard oscillator calculations give:
\begin{equation}
\label{eq:singular:vector:integrand}
\prod\limits_{s=1}^l \beta(z_s){\mathsf E}^{-\kappa^{-1}}(z_s)\widetilde{w}_{\kappa,j-l} 
= 
\prod\limits_{s<k}(z_s-z_k)^{2\kappa^{-2}}\prod\limits_{s=1}^l z_s^{2\kappa^{-2}(j-l+1)}\beta_{\leq}(z_s){\mathsf E}_{<}^{-\kappa^{-1}}(z_s)
\widetilde{w}_{\kappa,j},
\end{equation}
where:
\[
{\mathsf E}_{<}^{\alpha}(z) =  {\mathrm e}^{2\alpha\sum\limits_{m=1}^\infty \frac{a_{-m}}{m}z^m}
\hskip 1cm \mathrm{and} \hskip 1cm
\beta_{\leq}(z) = \sum\limits_{m=0}^{\infty}z^{m-1}\beta_{-m}.
\]
Since, up to the terms vanishing  for $\kappa^{-1}\to 0:$
\begin{eqnarray}
&&
\hskip -1.5cm
\oint\limits_{C_1}\!dz_1\ldots \hskip -3pt\int\limits_{C_l}\hskip -5pt dz_l\ 
\prod\limits_{s<k}(z_s-z_k)^{2\kappa^{-2}}\prod\limits_{s=1}^l z_s^{2\kappa^{-2}(j^+_{r,l}-l+1)}\beta_{\leq}(z_s){\mathsf E}_{<}^{-\kappa^{-1}}(z_s)\widetilde{w}_{\kappa,j^+_{r,l}}
\\
\nonumber
& \stackrel{\kappa^{-1}\to\, 0}{\approx} &
\oint\limits_{C_1}\!dz_1\ldots \hskip -3pt\int\limits_{C_l}\hskip -5pt dz_l\ 
\prod\limits_{s=1}^l z_s^{-r}\beta_{\leq}(z_s)\widetilde{w}_{\kappa,j^+_{r,l}} 
=
(2\pi i)^{l}\Bigg(\oint\limits_{0}\!\frac{dz}{2\pi i}\ z^{-r}\beta_{\leq}(z)\Bigg)^{l}\widetilde{w}_{\kappa,j^+_{r,l}},
\end{eqnarray}
we conclude that
\begin{equation}
\label{eq:chi:rl:through:Q:asymptotics}
 {\mathsf Q}_-^l\widetilde{w}_{\kappa,j^+_{r,l}-l}
=
(2\pi i)^l\beta_{-r}^l\widetilde{w}_{\kappa,j^+_{r,l}} + o(1)_{\kappa \to \infty},
\end{equation}
where $o(1)_{\kappa \to \infty}$ denotes terms vanishing for $\kappa \to \infty$. An analogous calculation give:
\begin{equation}
{\mathsf Q}_-^l  w_{\kappa,j^-_{r,l}+l}
=
(2\pi i)^l\beta_{-r}^l w_{\kappa,j^-_{r,l}} + o(1)_{\kappa \to \infty}.
\label{eq:chi:rl:through:Q:asymptotic:minus}
\end{equation}

We note that up to a multiplicative factor, formulas \eqref{eq:chi:rl:through:Q:asymptotics} and \eqref{eq:chi:rl:through:Q:asymptotic:minus} determine the asymptotics of $\left. s(\chi^-_{r,l}) \right|_{T^{-}_{r,l}=0}$ and $\left. \widetilde s(\chi^+_{r,l}) \right|_{T^{+}_{r,l}=0}$ for large $\kappa$. Another structural property of these vectors seen in the calculation with screening charges is that they may be obtained by acting exclusively with operators $a_n$ and $\beta_n$, without $\gamma_n$. This can also be seen directly, because $\left. s(\chi^-_{r,l}) \right|_{T^{-}_{r,l}=0}$ is annihilated by $J_0^+=\beta_0, J_1^+ = \beta_1$ etc., and analogously for $\left. \widetilde s(\chi^+_{r,l}) \right|_{T^{+}_{r,l}=0}$.

The limit $k \to \infty $ can also be interpreted in semiclassical terms. One way to make this systematic is with Poisson vertex algebras, see e.g.\ the textbook \cite{FBZ}. or the recent work \cite{BigCenter}. 

\section{Tensor product}
\label{section:tensor:product}

Let $\ut \V^{k,j, \epsilon} = \V^{k,j} \otimes_{\C} \H^{1, \epsilon}$. In this space we have two commuting sets of currents, $\sJ^a(z)$ inherited from $\V^{k,j}$, and $K^a(z)$ inherited from $\H^{1, \epsilon}$. The combined currents
\begin{equation}
    \J^a(z) = \sJ^a(z) +  K^a(z) = \sum_{n \in \Z} \frac{\mathcal J^a_n}{z^{n+1}}
    \label{eq:J_curl}
\end{equation}
make $\ut \V^{k,j, \epsilon}$ a representation of $\asl$ at level $k+1$. 

Suppose that $k \not \in \{ -2 , -3 \}$, which we assume for the remainder of this Section. Let $T^{J}, T^K$ and $T^{\J}$ be Sugawara fields constructed from the respective currents. Define
\begin{equation}
    T^{\Vir}(z) = T^J(z) + T^K(z) - T^{\J}(z) = \sum_{n \in \Z} \frac{L_n^{\Vir}}{z^{n+2}}.
    \label{eq:T_Vir}
\end{equation}
Then $T^{\Vir}(z)$ satisfies OPE of a Virasoro current with central charge 
\begin{equation}
 c_k^{\otimes} = c_k + c_1 - c_{k+1}.   
\end{equation}
Moreover, it commutes with $\J^a(z)$. Hence $\ut \V^{k,j,\epsilon}$ is a~representation of $\asl \oplus \Vir$, where $\Vir$ is the Virasoro algebra spanned by $L_n^{\Vir}$ \cite{GKO}.

First hints on the structure of $\ut \V^{k,j,\epsilon}$ regarded as a representation of $\asl \oplus \Vir$ can be obtained by inspection of characters. We have \cite{DiFrancesco:1997nk}:
\begin{eqnarray}
\nonumber 
\chi_{k,j}^{\V}(q,y) 
& := & 
\mathrm{Tr}_{\V^{k,j}}\,q^{L_0-\frac{\scriptstyle{c}_{\scriptscriptstyle k}}{24}}y^{J^0_0}
=
\frac{q^{-\frac{\scriptstyle{c}_{\tiny{k}}}{24}}}{1-y^{-1}}\frac{y^jq^{\frac{j(j+1)}{k+2}}}{(qy;q)(q;q)(qy^{-1};q)},
\\[4pt]
\chi_{1,\epsilon}^{\H}(q,y)
& := &
\mathrm{Tr}_{\mathcal{H}^{1,\epsilon}}\,q^{L_0-\frac{\scriptstyle{c}_{\scriptscriptstyle 1}}{24}}y^{J^0_0} 
=
\frac{q^{-\frac{\scriptstyle{c}_{\tiny{1}}}{24}}}{(q;q)}\sum\limits_{m \in \mathbb{Z}}q^{(m+\epsilon)^2}y^{m+\epsilon},
\end{eqnarray}
where 
\begin{equation}
(z;q) = \prod\limits_{n=1}^\infty \left(1-zq^{n-1}\right).
\end{equation}
Consequently,
\begin{equation}
\label{eq:chracter:decomposition}
\chi_{k,j}^{\V}(q,y)\chi_{1,\epsilon}^{\H}(q,y)
=
\sum\limits_{n\in\mathbb Z} \chi_{k+1,j+\epsilon + n}^{\V}(q,y)\chi^{\Vir}_{c_k^\otimes,\Delta^n_{k,j,\epsilon}}(q),
\end{equation}
where 
\begin{equation}
\Delta^n_{k,j,\epsilon} = \frac{j(j+1)}{k+2} -\frac{(j+\epsilon+n)(j+\epsilon+n+1)}{k+3} + (n+\epsilon)^2
\label{eq:highest_weights_Vir}
\end{equation}
and
\begin{equation}
    \chi^{\Vir}_{c,\Delta}(q) = \frac{q^{\Delta - \frac{c}{24}}}{(q;q)}
\end{equation}
is the character of the Virasoro Verma module $\V_{\Vir}^{c,\Delta}$ with the $L_0$-weight $\Delta$ and the central charge~$c.$ Identity \eqref{eq:chracter:decomposition} does not prove, but is suggestive of existence of a decomposition 
\begin{equation}
    \V^{k,j} \otimes \H^{1 , \epsilon} \cong \bigoplus_{n \in \Z} \V^{k+1,j+\epsilon+n} \otimes \V^{c_k^{\otimes},\Delta^n_{k,j,\epsilon}}_{\Vir}.
    \label{eq:module:decomp}
\end{equation}

The goal of the present Subsection is to explicitly construct vectors $v^n_{\kappa,j,\epsilon} \in \V^{k,j,\epsilon}$, $n \in \Z$, which are highest weight vectors of $\asl{}_{,k+1} \oplus \Vir$ with weights as indicated in
\eqref{eq:module:decomp}.

Let $\ut \W^{\kappa,j, \epsilon} = \W^{\kappa,j} \otimes_{\C} \H^{1,\epsilon}$ and $\ut \Wt^{\kappa,j, \epsilon} = \Wt^{\kappa,j} \otimes_{\C} \H^{1,\epsilon}$. Constructions (\ref{eq:J_curl},\ref{eq:T_Vir}) may be repeated in these modules. Homomorphisms $s, \widetilde s$ induce maps:
\begin{equation}
     \mathbf s = s \otimes 1_{\H^{1, \epsilon}}: {}^{\otimes}\V^{k,j, \epsilon} \to \ut \W^{\kappa,j, \epsilon}, \qquad  \widetilde {\mathbf s} = \widetilde s \otimes 1_{\H^{1, \epsilon}} : {}^{\otimes}\V^{k,j, \epsilon} \to \ut \Wt^{\kappa,j, \epsilon}
\end{equation}
satisfying
\begin{equation}
\label{eq:relation_between_bilinear_forms}
    g(v,v' ) =  \langle  \widetilde {\mathbf s}(v) | \mathbf s (v') \rangle ,
\end{equation}
where $g$, $ \langle \cdot | \cdot \rangle$ are the self-evident induced pairings on $\ut \V^{k, j , \epsilon}$ and between $\ut \Wt^{\kappa, j , \epsilon}$ and $\ut \W^{\kappa,j,\epsilon}$, respectively. We hope that the overloaded notation is not too confusing.

We define the following fields:
\begin{align}
    \rho(z) &= \gamma^2(z) K^+(z) - 2 \gamma(z) K^0(z) - K^-(z) - \partial \gamma(z) \qquad && \text{in } \ut \W^{\kappa,j,\epsilon}, \\
    \widetilde \rho(z) &=  \gamma^2(z) K^-(z) - 2 \gamma(z) K^0(z) - K^+(z) + \partial \gamma(z) \qquad && \text{in } \ut \Wt^{\kappa,j,\epsilon}. \nonumber
\end{align}

\begin{proposition} \label{prop:rho_OPE}
The following OPE are satisfied:
\begin{align}
    & \J^+(z) \rho(w) \sim \mathrm{reg}, \qquad && \J^+(z) \widetilde \rho(w) \sim \frac{2 \gamma(w) \widetilde \rho(w)}{z-w}, \nonumber \\
    & \J^0(z) \rho(w) \sim - \frac{\rho(w)}{z-w}, \qquad && \J^0(z) \widetilde \rho(w) \sim \frac{\widetilde \rho(w)}{z-w}, \nonumber \\
    & \J^-(z) \rho(w) \sim \frac{2 \gamma(w) \rho(w)}{z-w}, \qquad && \J^-(z) \widetilde \rho(w) \sim \mathrm{reg}, \label{eq:rho_OPE} \\
   &  (T^J(z)+T^K(z))\rho(w) \sim \frac{\partial}{\partial w} \frac{\rho(w)}{z-w}, && (T^J(z)+T^K(z))\widetilde \rho(w) \sim \frac{\partial}{\partial w} \frac{\rho(w)}{z-w}, \nonumber \\
    & \rho(z) \rho(w) \sim \mathrm{reg}, \qquad && \widetilde \rho(z) \widetilde \rho(w) \sim \mathrm{reg}. \nonumber
\end{align}
\end{proposition}
\begin{proof}
We sketch the computation for $\rho$, the case of $\widetilde \rho$ being analogous. OPE of $J$ and $K$ with $\rho$ are easy to get, but separately not very illuminating -- cancellations occur upon adding them up. To calculate the OPE of $T^J+T^K$ with $\rho$, we use \eqref{eq:free_T} to expand $T^J(z) \rho(w)$ and the OPE satisfied by $K^a$ to expand $T^K(z) \rho(w)$. The OPE between $\rho$ and $\rho$ is easy.
\end{proof}

Let us write
\begin{equation}
    \rho(z) = \sum_{n \in \Z} \frac{\rho_n}{z^{n+1}}, \qquad \widetilde \rho(z) = \sum_{n \in \Z} \frac{\widetilde \rho_n}{z^{n+1}}.
\end{equation}
Many useful commutation rules obeyed by $\rho_n, \widetilde \rho_n$ are encoded in Proposition \ref{prop:rho_OPE}, e.g.
\begin{equation}
    \label{eq:comm:LVir:rho}
  [\mathcal J^+_n, \rho_m]=0, \qquad  [\mathcal J^0_n,\rho_m] = -\rho_{n+m}, \qquad  [L^J_n+L^K_n,\rho_m] = -m\rho_{m+n}.
\end{equation}
We note also the identity
\begin{equation}
    \langle \varphi | \rho_n w \rangle = \langle \widetilde \rho_{-n} \varphi | w \rangle \qquad \text{for } \varphi \in \ut \Wt^{\kappa,j,\epsilon}, \, w \in \ut \W^{\kappa,j,\epsilon}.
\end{equation}

We proceed to construct highest weight vectors in $\ut \W^{\kappa,j,\epsilon}$ and $\ut \widetilde \W^{\kappa,j,\epsilon}$. We start from the most basic ones:
\begin{equation}
    w_{\kappa,j, \epsilon}^0 = w_{\kappa,j} \otimes f_\epsilon \in \ut \W^{\kappa,j,\epsilon}, \qquad \widetilde w_{\kappa,j,\epsilon}^0 = \widetilde w_{\kappa,j} \otimes f_{\epsilon} \in \ut \widetilde \W^{\kappa,j, \epsilon}.
\end{equation}

Then, for any $n \geq 1$ we introduce:
\begin{align}
\label{eq:special:states}
\nonumber
& w_{\kappa,j,0}^n = \rho_{-2n+1} \cdots \rho_{-3}\rho_{-1} w_{\kappa,j,0}^0 \hskip 30pt\in\; \ut \W^{\kappa,j,0},
\\[2pt]
\nonumber
& w_{\kappa,j,\frac12}^n = {(-1)^n}\rho_{-2n+2} \cdots \rho_{-2}\rho_0 w_{\kappa,j,\frac12}^0 \hskip 3pt\in\; \ut \W^{\kappa,j,\frac12},
\\[-7pt]
\\[-7pt] \nonumber  
& \widetilde w_{\kappa,j,0}^n = \widetilde \rho_{-2n+1} \cdots \widetilde\rho_{-3}\widetilde\rho_{-1}\widetilde w_{\kappa,j,0}^0  \hskip 30pt \in\; \ut \Wt^{\kappa,j,0},
\\[2pt]
\nonumber
& \widetilde w_{\kappa,j,\frac12}^n = (-1)^n \widetilde \rho_{-2n} \cdots \widetilde\rho_{-2}\widetilde w_{\kappa,j,\frac12}^0 \hskip 25pt\in\; \ut \Wt^{\kappa,j,\frac12}.
\end{align}

To prove that vectors \eqref{eq:special:states} have the desired properties (Proposition \ref{prop:explicit:hs:states} below), it is useful to use the embedding of $\H^{1, \epsilon}$ in $\F$ explained in Proposition \ref{prop:Fermi_Rep}. This allows to express $\rho(z)$ and $\widetilde{\rho}(z)$ fields (acting in extended spaces $\W^{\kappa, j} \otimes \F$ and $\widetilde \W^{\kappa, j} \otimes \F$) as products of fermionic fields:
\begin{equation}
\label{eq:rho:through:fermions}
    \rho(z) = \chi(z){\overline \chi}(z),
    \hskip 1cm
    \widetilde{\rho}(z) = \widetilde{\chi}(z)\widetilde{\overline \chi}(z),
\end{equation}
where
\begin{equation*}
  \chi(z) = \psi^1(z)- \gamma(z)\psi^2(z), \hskip 5mm \overline{\chi}(z) = \psibar^2(z)+\gamma(z)\psibar^1(z),
\end{equation*}
and
\begin{equation*}
  \widetilde{\chi}(z) = \psi^2(z)+\gamma(z)\psi^1(z), \hskip 5mm \widetilde{\overline{\chi}}(z) = \psibar^1(z)- \gamma(z)\psibar^2(z).
\end{equation*}

Modes of $\chi$ are defined by
\begin{equation}
    \chi(w) = \sum\limits_{k\in{\mathbb Z}+\frac12}\frac{\chi_k}{z^{k+\frac12}},
\end{equation}
and analogously for $\overline \chi, \widetilde \chi$, and $\widetilde{\overline \chi}$. Modes of $\chi_k$ and $\widetilde \chi_k$ all anticommute, and the same is true for $\widetilde \chi_k$ and $\widetilde{\overline \chi}_k$. That is, we have superbrackets:
\begin{equation}
\label{eq:chi:superbrackets}
[\chi_k,\chi_l] = [\overline{\chi}_k,\chi_l] = [\overline{\chi}_k,\overline{\chi}_l]
=
[\widetilde{\chi}_k,\widetilde{\chi}_l] = [\widetilde{\overline{\chi}}_k,\widetilde{\chi}_l] = [\widetilde{\overline{\chi}}_k,\widetilde{\overline{\chi}}_l]
=
0.
\end{equation}
Vector $w^0_{\kappa,j,\frac12}$ (resp.\ $\widetilde w^0_{\kappa,j,0}$) is annihilated by $\chi_k$ for $k \geq \frac32$ and by $\chibar_k$ for $k \geq \frac12$ (resp. by $\widetilde \chi_k$ and $\widetilde \chibar_k$ for $k \geq \frac12$). All the other $w^n_{\kappa,j,\epsilon}$ and $\widetilde w^n_{\kappa,j,\epsilon}$ may be obtained recursively by acting with $\chi$ and $\chibar$ (resp. $\widetilde \chi$ and $\widetilde \chibar$):
\begin{align}
    w^n_{\kappa,j,0} =&\; \chi_{- \frac{2n-1}{2}} w_{\kappa, j ,\frac12}^n, \qquad w^{n+1}_{\kappa,j,\frac12}= \overline \chi_{- \frac{2n+1}{2}} w_{\kappa,j,0}^n, \\
    \widetilde w_{\kappa, j , \frac12}^n =&\; \widetilde{\overline \chi}_{- \frac{2n+1}{2}} \widetilde w_{\kappa,j,0}^n, \qquad\hskip 4pt \widetilde w_{\kappa,j,0}^{n+1} = \widetilde \chi_{- \frac{2n+1}{2}} \widetilde w_{\kappa,j,\frac12}^n.
\end{align}
More explicitly, we have
\begin{align}
\label{eq:special:states:through:fermions}
w_{\kappa,j,\frac12}^{n}  = &\; \chibar_{-\frac{2n-1}{2}} \chi_{- \frac{2n-3}{2}} \cdots \chibar_{-\frac12} \chi_{\frac12} w_{\kappa,j,\frac12}^0, 
\hskip 30pt 
w_{\kappa,j,0}^n  =\; \chi_{-\frac{2n-1}{2}}  w_{\kappa,j, \frac12}^n, 
\\[2pt]
\nonumber
\widetilde w_{\kappa,j,0}^n  = &\; \widetilde{\chi}_{-\frac{2n-1}{2}}\widetilde{\chibar}_{-\frac{2n-1}{2}} \cdots \widetilde{\chi}_{-\frac12}\widetilde{\chibar}_{-\frac12}\widetilde w_{\kappa,j,0}^0, 
\qquad 
\widetilde w_{\kappa,j,\frac12 }^{n}  =\; \widetilde \chibar_{- \frac{2n+1}{2}} \widetilde w_{\kappa,j,0}^n.
\end{align}

\begin{proposition}
\label{prop:explicit:hs:states}
$w_{\kappa,j,\epsilon}^n, \widetilde w_{\kappa,j,\epsilon}^n$ are nonzero vectors annihilated by $\{ \J^a_n ,L_n^{\Vir} \}_{n >0}$ and $\J_0^+$. They satisfy eigenvalue equations
\begin{align}
   & \J^0_0 w_{\kappa,j,\epsilon}^n = (j + \epsilon -n) w_{\kappa,j,\epsilon}^n, \quad && L_0^{\Vir} w_{\kappa,j,\epsilon}^n = \Delta_{k,j, \epsilon}^{-n}  w_{\kappa,j,\epsilon}^n, \label{eq:tensor_hw_eig} \\
    &\J^0_0 \widetilde w_{\kappa,j,\epsilon}^n = (j+\epsilon +n) \widetilde w_{\kappa,j,\epsilon}^n, \quad && L_0^{\Vir} \widetilde w_{\kappa,j,\epsilon}^n = \Delta_{k,j, \epsilon}^n  \widetilde w_{\kappa,j,\epsilon}^n, \nonumber
\end{align}
where $\Delta_{k,j,\epsilon}^n$ are as in \eqref{eq:highest_weights_Vir}. 
\end{proposition}
\begin{proof} 
We discuss vectors $w_{\kappa,j,0}^n$, the other three being analogous. If $n=0$, the statement is obvious. Now let $n$ be general. By \eqref{eq:comm:LVir:rho}, $w_{\kappa,j,0}^n$ is annihilated by $\mathcal J^+_m$ for $m \geq 0$, by $\mathcal J^0_m$ for $m>0$, and by $\mathcal J_0^0-j+n$. To argue that it is also annihilated by $\mathcal J^-_m$ for $m > 0$, we use the fermionic formula \eqref{eq:special:states:through:fermions}. We have commutators
\begin{equation}
    [\mathcal J_m^-,\chi_{-k}] = \sum_{l \in \Z} \gamma_l \chi_{-k-l+m}, \qquad [\mathcal J_m^-,\chibar_{-k}] = \sum_{l \in \Z} \gamma_l \chibar_{-k-l+m}.
\end{equation}
By Leibniz rule for commutators, $\mathcal J^-_m w_{\kappa,j,0}^n$ is the sum of terms of the form
\begin{equation}
    \pm \sum_{l \in \Z} \left( \prod_{i=1}^n \chibar_{- \frac{2i-1}{2}} \right) \left( \prod_{i \neq i_0} \chi_{- \frac{2i-1}{2}} \right)  \gamma_l \chi_{- \frac{2i_0-1}{2} -l+m} w^{0}_{\kappa,j,0}
\end{equation}
and analogous with the roles of $\chi$ and $\chibar$ interchanged. In the sum over $l$, terms with $l > 0$ vanish because $\gamma_l w^{0}_{\kappa,j,0} =0$, and all other terms vanish because $\chi_{- \frac{2i_0-1}{2} -l+m}$ either annihilates $w^{0}_{\kappa,j,0}$ or is included in the product $\prod_{i \neq i_0} \chi_{- \frac{2i-1}{2}}$ (and $\chi_{-k}^2 =0$). This result also implies, that
\begin{equation}
\label{eq:lm:for:sum:being:zero}
L^{\J}_m w_{\kappa,j,0}^n = 0, \hskip 1cm m > 0.
\end{equation}

Using nilpotency of $\chi_{-k}$ and $\chibar_{-k}$ one also verifies that for any $l \geq 0:$
\begin{equation}
\rho_{-2l}w_{\kappa,j,l}^0
=
\rho^2_{-2l-1}w_{\kappa,j,l}^0
=
0
\end{equation}
so that, using \eqref{eq:comm:LVir:rho} and \eqref{eq:lm:for:sum:being:zero}, we obtain the Virasoro highest weight condition
\begin{equation}
    0 = \left(L^J_m + L^K_m\right)w_{\kappa,j,0}^n = L^{\Vir}_mw_{\kappa,j,0}^n, \hskip 1cm m > 0. 
\end{equation}

Finally we verify the eigenequation for $L_0^{\Vir}$. By \eqref{eq:comm:LVir:rho} we have
\begin{equation}
    (L_0^{J}+L_0^K) w^{n}_{\kappa,j, 0} = \left( \Delta_{k,j} + \sum_{i=1}^{n} (2i-1) \right) w^{n}_{\kappa,j, 0}.
\end{equation}
Since $w^n_{\kappa,j,0}$ is a highest weight vector for $\mathcal J^a$, we have also
\begin{equation}
L_0^{\mathcal J} w^n_{\kappa,j,0} = \Delta_{k+1,j-n} w^n_{\kappa,j,0}.
\end{equation}
\end{proof}

\begin{lemma} \label{lem:overlap_calc}
    We have
    \begin{equation}
     \left.   \langle  \widetilde s( T \chi^+_{r,l}) \otimes f , w^n_{\kappa,j,\epsilon} \rangle \right|_{T^+_{r,l}} \equiv 0, 
     \hskip 5mm \text{resp.} \hskip 5mm
     \left.   \langle   \widetilde w^n_{\kappa,j,\epsilon}, s( T \chi^-_{r,l}) \otimes f \rangle \right|_{T^-_{r,l}} \equiv 0,
     \label{eq:spec_state_overlap_singular}
    \end{equation}
    for all $T \in U \aslm$ and all $ f \in \H^{1,\epsilon}$ if and only if $r+l > 2n - 2 \epsilon$ (resp. $r+l > 2n+2\epsilon$).
\end{lemma}
\begin{proof}
Consider the left hand side of the first expression in \eqref{eq:spec_state_overlap_singular} (the analysis of the second expression is analogous). $T$ is a polynomial in $J^a_{-m}$ ($m<0$ or $m=0$, $a=-$). We may use adjoint relations for $\sJ^a_m, K^a_m$ and the equations $J^a_m w^n_{\kappa,j,\epsilon} = -K^a_m w^n_{\kappa,j,\epsilon}$ to rewrite \eqref{eq:spec_state_overlap_singular} in the same form with $T=1$ and some new $ f$. As discussed in Subsection \ref{sec:screenieng_charges},  $\left. \widetilde s(  \chi^+_{r,l}) \right|_{T^+_{r,l}=0}$ is a linear combination of terms of the form 
\begin{equation}
    a_{-j_1} \cdots a_{-j_p} \beta_{-i_1} \cdots \beta_{-i_l} \widetilde w_{\kappa,j}
\end{equation}
with $j_1 + \dots + j_p +i_1 + \dots + i_l = rl$. Terms with $p \neq 0$ do not contribute in \eqref{eq:spec_state_overlap_singular} because $w^n_{\kappa,j , \epsilon}$ does not contain any $a$ operators (and thus is annihilated by $a_n$ for $n>0$). Therefore, \eqref{eq:spec_state_overlap_singular} is a linear combination of terms of the form
\begin{equation}
    \langle \beta_{-i_1} \cdots \beta_{-i_l} \widetilde w_{\kappa,j} \otimes f , w^n_{\kappa,j,\epsilon} \rangle, 
    \label{eq:overlap_term}
\end{equation}
with $i_1 + \dots + i_l =rl$. Next we note that $\beta_n = \widetilde J^-_n$, so using adjointness relations and the highest weight condition again we transform the term \eqref{eq:overlap_term} to form
\begin{equation}
    \langle \widetilde w_{\kappa,j} \otimes K_{-i_1}^- \cdots K_{-i_l}^- f , w^n_{\kappa,j,\epsilon} \rangle,
\end{equation}
possibly with a sign absorbed into a redefinition of $f$. In this expression we can drop all terms in $w^n_{\kappa,j,\epsilon}$ containing any $\gamma$ operators. Therefore, by definition of $w^n_{\kappa,j,\epsilon}$, we are left with (again we absorb signs)
\begin{equation}
    \langle K_{-i_1}^- \cdots K_{-i_l}^- f |  \prod_{i=1}^n \psi^1_{- \frac{2i-1}{2} +2 \epsilon} \overline \psi^2_{- \frac{2i-1}{2}} \cdot f_\epsilon \rangle.
\end{equation}
Now we use the adjointness relations for $K^-$ and then express $K^+$ in terms of fermions to rewrite this expression as a combination of terms
\begin{equation}
    \langle f | \overline \psi^1_{a_1} \psi^2_{b_1} \cdots \overline \psi^1_{a_l} \psi^2_{b_l} \cdot \prod_{i=1}^n \psi^1_{- \frac{2i-1}{2} +2 \epsilon} \overline \psi^2_{- \frac{2i-1}{2}} \cdot f_\epsilon \rangle,
    \label{eq:random_label}
\end{equation}
where the indices $a_j,b_j$ satisfy $\sum_{j=1}^l (a_j + b_j) = rl$. Clearly the expression \eqref{eq:random_label} vanishes if any $a_j > \frac{2n-1}{2} - 2 \epsilon$ or $b_j > \frac{2n-1}{2}$. Moreover, it vanishes if there are repetitions among $a_j, b_j$. Hence for \eqref{eq:random_label} to be nonzero it is necessary that we have
\begin{equation}
    \sum_j a_j \leq \sum_{j=0}^{l-1} \left( \frac{2n-1}{2} -2  \epsilon-j \right) = l \left( \frac{2n-l}{2} -2 \epsilon \right).
\end{equation}
Similarly we derive $\sum_j b_j \leq l  \frac{2n-l}{2}$, and hence
\begin{equation}
    rl = \sum_{j=1}^l(a_j + b_j) \leq l (2n-l - 2 \epsilon),
\end{equation}
which is equivalent to $r+l \leq 2n - 2 \epsilon$. This completes the proof of the first part of the statement.

Suppose now that $r+l\leq 2n-2\epsilon.$ To construct an example of a pair $(T, f)\in \aslm\times \H^{1,\epsilon}$ such that 
$\left.\langle  \widetilde s( T \chi^+_{r,l}) \otimes f , w^n_{\kappa,j,\epsilon} \rangle \right|_{T^+_{r,l}}$ is not identically zero it is enough to take $T=1.$
Formula (\ref{eq:chi:rl:through:Q:asymptotics}) implies that there exists a polynomial $C_{r,l}(\kappa)$ such that
\begin{equation}
    \left.   \langle \widetilde{s}\left(\chi^+_{r,l} \right)  \otimes f , w^n_{\kappa,j,\epsilon} \rangle \right|_{T^+_{r,l}=0}
    = C_{r,l}(\kappa) \left( \left. \langle \beta_{-r}^l\widetilde{w}_{\kappa,j^+_{r,l}} \otimes f , w^n_{\kappa,j,\epsilon} \rangle \right|_{T^+_{r,l}=0} + R_{r,l}(\kappa) \right)
\end{equation}
where $\lim\limits_{\kappa\to\infty}R_{r,l}(\kappa) = 0.$  
A direct calculation  shows that if we take
\begin{equation}
    f = \left\{
    \begin{array}{rcl}
    \frac{(-1)^n}{r!} K^-_{-2n+1}\ldots K^-_{-2r-3}K^-_{-2r-1} \left(K^-_{-r}\right)^{r-l}f_0 & \;\mathrm{for}\; & \epsilon = 0,\; r \geq l,
    \\[6pt]
     \frac{(-1)^n}{l!} K^-_{-2n+1}\ldots K^-_{-2l-1} K^0_{-2l+r+1}\ldots K^0_{-r-3}K^0_{-r-1} f_0& \;\mathrm{for}\; & \epsilon = 0,\; r < l,
    \\[6pt]
    \frac{1}{(r-1)!} K^-_{-2n+2}\ldots K^-_{-2r-2}K^-_{-2r}\left(K^-_{-r}\right)^{r-l-1}K^-_0f_{\frac12} & \;\mathrm{for}\; & \epsilon = \frac12,\; r > l,
    \\[8pt]
   \frac{1}{l!}K^-_{-2n+2}\ldots K^-_{-2l-2}K^0_{-2l+r}\ldots K^0_{-r-2}K^0_{-r}K^-_0f_\frac12  & \;\mathrm{for}\; & \epsilon = \frac12,\; r \leq l,
    \end{array}
    \right.
\end{equation}
then
\begin{equation}
    \left. \langle \beta_{-r}^l\widetilde{w}_{\kappa,j^+_{r,l}} \otimes f , w^n_{\kappa,j,\epsilon} \rangle \right|_{T^+_{r,l}} = 1.
\end{equation}
Indeed,
\begin{equation}
\langle \beta_{-r}^l\widetilde{w}_{\kappa,j^+_{r,l}} \otimes f , w^n_{\kappa,j^+_{r,l},0} \rangle
=
\langle \widetilde{w}_{\kappa,j^+_{r,l}} \otimes \left(K^-_{-r}\right)^l f , w^n_{\kappa,j^+_{r,l},0} \rangle, 
\end{equation}
and for $\epsilon = 0,\ r \geq l:$
\begin{eqnarray}
\nonumber
&& \hskip -3cm
\langle \widetilde{w}_{\kappa,j^+_{r,l}} \otimes \left(K^-_r\right)^l f , w^n_{\kappa,j^+_{r,l},0} \rangle
=
\frac{(-1)^n}{r!}\langle \widetilde{w}_{\kappa,j^+_{r,l}} \otimes K^-_{-2n+1}\ldots K^-_{-2r-1} \left(K^-_{-r}\right)^rf_0 , w^n_{\kappa,j^+_{r,l},0} \rangle
\\[4pt]
& = &
\langle \widetilde{w}_{\kappa,j^+_{r,l}} \otimes K^-_{-2n+1}\ldots K^-_{-3} K^-_{-1}f_0 , w_{\kappa,j^+_{r,l}} \otimes K^-_{-2n+1}\ldots K^-_{-3} K^-_{-1}f_0 \rangle
= 1
\end{eqnarray}
while for $\epsilon = 0,\ r < l:$
\begin{eqnarray}
\nonumber
&& \hskip -1cm
\langle \widetilde{w}_{\kappa,j^+_{r,l}} \otimes \left(K^-_{-r}\right)^l f , w^n_{\kappa,j^+_{r,l},0} \rangle
\\
& = &
\frac{(-1)^n}{l!}\langle \widetilde{w}_{\kappa,j^+_{r,l}} \otimes 
 K^-_{-2n+1}\ldots K^-_{-2l-1} \left(K^-_{-r}\right)^lK^0_{-2l+r+1}\ldots K^0_{-r-3}K^0_{-r-1} f_0 , w^n_{\kappa,j^+_{r,l},0} \rangle
\\[4pt]
\nonumber
& = &
\frac{(-1)^n}{r!}\langle \widetilde{w}_{\kappa,j^+_{r,l}} \otimes 
 K^-_{-2n+1}\ldots K^-_{-2l-1}K^-_{-2l+1}\ldots K^-_{-2r-3}K^-_{-2r-1}  \left(K^-_{-r}\right)^rf_0 , w^n_{\kappa,j^+_{r,l},0} \rangle
= 1,
\end{eqnarray}
where we used the fact that $(K_{-r}^-)^{r+1} f_0 =0$. Computations for $\epsilon = \frac12$ are similar.
\end{proof}

Now for any $n \in \Z$ we define 
\begin{equation}
v^n_{\kappa,j,\epsilon} = \begin{cases}
         \prod_{r+l \leq 2|n| - 2 \epsilon} T^+_{r,l} \cdot \mathbf s^{-1}(w^{|n|}_{\kappa,j,\epsilon}) & \text{if } n \leq 0, \\[4pt]
         \prod_{r+l \leq 2n + 2 \epsilon} T^-_{r,l} \cdot \widetilde{\mathbf s}^{-1}(\widetilde w^{n}_{\kappa,j,\epsilon}) & \text{if } n \geq 0.
    \end{cases}
\label{eq:special_states_Verma}
\end{equation}
Then $v^n_{\kappa,j,\epsilon}$ depends polynomially on $\kappa,j$ by Lemmas \ref{lem:overlap_criterion}, \ref{lem:kappa_singularity}, \ref{lem:overlap_calc}. If $n=0$, expressions in the two cases of \eqref{eq:special_states_Verma} are both equal $v_{k,j} \otimes  f_{\epsilon}$.

\begin{proposition}
Let $n \neq 0$. There exists at most finitely many pairs $\kappa,j$ such that $v^n_{\kappa,j,\epsilon}=0$. They are all contained in the union of zero loci of $\kappa$ and $(T^+_{r,l})_{r+l \leq 2|n| - 2 \epsilon }$ ($n < 0$) or $\kappa$ and $(T^-_{r,l})_{r+l \leq 2n + 2 \epsilon }$ ($n > 0$).
\label{prop:renormalization}
\end{proposition}
\begin{proof}
We discuss $n > 0$. We have $\widetilde{\mathbf s} (v^n_{\kappa,j,\epsilon}) = \prod_{r+l \leq 2n + 2 \epsilon} T^-_{r,l} \cdot\widetilde w^{n}_{\kappa,j,\epsilon}$, which is nonzero if $\kappa$ and all $T^-_{r,l}$ in the product are nonzero. Therefore, the zero locus of $v^n_{\kappa,j,\epsilon}$ is contained in a finite union of complex lines. By construction, $v^n_{\kappa,j,\epsilon}$ does not vanish identically on any of these lines, so its zero locus is zero-dimensional. 
\end{proof}
We have, for example,
\begin{eqnarray}
\label{eq:highest:weight:tensor:product:1}
v^{-1}_{\kappa,j,\frac12}
& = & 
T^+_{0,1}\,{\mathbf s}^{-1}\left(w^1_{\kappa,j,\frac12}\right)
=
\left(-\sJ^-_0 \otimes {\mathbf 1} +2j\, {\mathbf 1}\otimes K^-_0\right) (v_{k,j}\otimes  f_{\frac12}),
\\[6pt]
\nonumber
v^1_{\kappa,j,0}
& = & 
T^-_{1,1}\,{\widetilde{\mathbf s}}^{-1}\left({\widetilde w}^1_{\kappa,j,0}\right) 
=
\left(-\sJ^+_{-1}\otimes {\mathbf 1} + (k-2j){\mathbf 1}\otimes  K^+_{-1}\right)(v_{k,j}\otimes  f_0),
\end{eqnarray}
as well as
\begin{eqnarray}
\label{eq:highest:weight:tensor:product:2a}
\nonumber
&&
\hskip -14mm
v^{-1}_{\kappa,j,0}
=
T^+_{0,1}T^+_{0,2}T^+_{1,1}\,{\mathbf s}^{-1}\left(w^1_{\kappa,j,0}\right) 
=
\left(-\sJ^+_{-1}\left(\sJ^-_0\right)^2 - 2(2j-1)\sJ^0_{-1}\sJ^-_0 +2j(2j-1)\sJ^-_{-1}\right)v_{k,j}\otimes  f_0
\\
&& \hskip -9mm
+ (k+2j+2)\left(\left(\sJ^-_0\right)^2\otimes  K^+_{-1} + 2(2j-1)\sJ^-_0\otimes  K^0_{-1} -2j(2j-1){\mathbf 1}\otimes  K^-_{-1}\right)(v_{k,j}\otimes  f_0),
\end{eqnarray}
and
\begin{eqnarray}
\label{eq:highest:weight:tensor:product:2b}
v^1_{\kappa,j,\frac12}
& = & 
T^-_{1,1}T^-_{1,2}T^-_{2,1}{\widetilde{\mathbf s}}^{-1}\left({\widetilde w}^1_{\kappa,j,\frac12}\right)
\\
\nonumber
& = & 
\left(-\sJ^-_{0}\left(\sJ^+_{-1}\right)^2 + 2(k-2j-1)\sJ^0_{-1}\sJ^+_{-1} + (k-2j)(k-2j-1)\sJ^+_{-2}\right)(v_{k,j}\otimes  f_{\frac12})
\\
\nonumber
&& 
+\;2(k-j+1)
\left(\left(\sJ^+_{-1}\right)^2\otimes K^-_0 -2(k-2j-1)\sJ^+_{-1}\otimes K^0_{-1}\right.
\\
\nonumber
&&
\left.
\hskip 5.8cm -(k-2j)(k-2j-1)\mathbf{1}\otimes K^+_{-2}\right)(v_{k,j}\otimes  f_{\frac12}),
\end{eqnarray}
in agreement\footnote{Recall that there are no singular vectors corresponding to 
$T^\pm_{r,0}$ and $T^-_{0,l},$ and in effect those factors are absent in (\ref{eq:special_states_Verma}) and (\ref{eq:highest:weight:tensor:product:1}) -- (\ref{eq:highest:weight:tensor:product:2b}).} 
with (\ref{eq:special_states_Verma}).

Let us note that in all examples above, $v^n_{\kappa,j,\epsilon}$ has no zeros: it contains a term with coefficient $1$. We have not been able to verify whether this is true in general. Nevertheless, we can deduce the following existence result for highest weight vectors:

\begin{corollary}
\label{cor:existence:of:hws}
    For any $k\not\in\{-2,-3\},$ any $j, \epsilon$ and for $n \in \Z$ there exists nonzero $v \in \ut \V^{k, j , \epsilon}$ which is annihilated by $\{ \mathcal J^a_n, L_n^{\Vir} \}_{n >0}$ and $\mathcal J_0^+$, and satisfies
    \begin{equation}
        \mathcal J^0_0 v = (j + \epsilon +n)v, \qquad L_0^{\Vir} v = \Delta_{k,j,\epsilon}^{n} v.
    \end{equation}
\end{corollary}
\begin{proof}
Fix $\kappa,j$. For some $m \in \N$ we have $v^n_{\kappa,j',\epsilon} = (j'-j)^m v_{j}$ with $v_{j'} \neq 0$. Take $v=v_j$.  
\end{proof}

\begin{remark}
If our conjecture that $v_{\kappa,j,\epsilon}$ has no zeroes is false, then one can't construct a~polynomial family of highest weight vectors by dividing out a normalization factor. The essence of the proof of Corollary \ref{cor:existence:of:hws} is to restrict to a one parameter family of vectors (varying $j$ with fixed $\kappa$). Approaching the point $\kappa,j$ from a different direction one would obtain different limits. Let us illustrate this phenomenon with an example. 

Consider the family of vectors $\xi(x,y) = \begin{bmatrix}
    x \\ y
\end{bmatrix}$ in $\mathbb C^2$ depending on two complex parameters $x,y$. We have $\xi(0,0)=0$, but there is no polynomial $p(x,y)$ such that $\frac{1}{p(x,y)} \xi(x,y)$ has a finite and nonzero limit at $(0,0)$. The situation improves if we restrict to one-dimensional subspaces, e.g.:
\begin{equation}
    \lim_{x \to 0} \frac{1}{x} \xi(x,0) = \begin{bmatrix}
        1 \\ 0
    \end{bmatrix}, \qquad  \lim_{y \to 0} \frac{1}{y} \xi(0,y) = \begin{bmatrix}
        0 \\ 1
    \end{bmatrix}.
\end{equation}
Note that these limits are not equal. 
\end{remark}

We can summarize this Subsection with the following Theorem.

\begin{theorem}
\label{th:character:decomposition}
An isomorphism \eqref{eq:module:decomp} has been constructed explicitly for generic $k,j$. In~particular we have the character identity
\begin{equation}
    \mathrm{Tr}_{\V^{k,j} \otimes \H^{1 , \epsilon}} \, \left( q_{\J}^{L_0^{\J} - \frac{c_{k+1}}{24}} y^{\J^0_0} q_{\Vir}^{L_0 - \frac{c_k^\otimes}{24}} \right) = \sum_{n \in \Z} \chi^{\V}_{k+1,j+\epsilon+n}(q_{\J},y) \chi^{\Vir}_{c_k^\otimes,\Delta_{k,j,\epsilon}^n}(q_{\Vir}).
\label{eq:finer_character}
\end{equation}
\end{theorem}
\begin{proof}
A module homomorphism
\begin{equation}
\bigoplus_{n \in \Z} \left( \V^{k+1,j+\epsilon + n} \otimes \V^{c_k^\otimes, \Delta_{k,j,\epsilon}^n}_{\Vir} \right) \to \V^{k,j} \otimes \H^{1 , \epsilon} 
\label{eq:constructed_homomorphism}
\end{equation}
is constructed by taking the highest weight vector of the $n$th summand of the domain to $v^n_{\kappa,j,\epsilon}$. For generic $\kappa,j,\epsilon$ the summands are irreducible and non-isomorphic, and $v^n_{\kappa,j,\epsilon} \neq 0$. Hence the homorphism is generically injective. Whenever it is injective, it is also surjective by \eqref{eq:chracter:decomposition}. For such $k,j$ equality \eqref{eq:finer_character} is obvious. It extends to general $k,j$ by taking limits.
\end{proof}

We note that \eqref{eq:finer_character} is a refinement of \eqref{eq:chracter:decomposition}: it reduces to it if we put $q_{\J}=q_{\Vir}=q$. We would also like to highlight that the assumption of $k,j$ being generic is essential to conclude the isomorphism \eqref{eq:module:decomp}. 
It is not enough to compare the characters (example: the Verma module and the Wakimoto module for $\asl$ are not isomorphic for certain weights, despite having a~highest weight vector of the same weight and equal characters; in such cases the highest weight vector is not cyclic in the Wakimoto module).

\appendix

\section{Noncommutative localization} \label{app:localization}

Let $R$ be a ring and let $S \subset R$ be a multiplicative set, i.e.\ $1 \in S$ and $s_1s_2 \in S$ for any two $s_1,s_2 \in S$. A ring homomorphism $j : R \to R_S$ (or often the ring $R_S$ itself, with $j$ implicit) is called the localization of $R$ with respect to $S$ if the following two conditions are satisfied:
\begin{itemize}
    \item $j(S)$ consists of invertible elements of $R_S$,
    \item if $ j' : R \to R'$ is a ring homomorphism such that $j'(S)$ consists of invertible elements, there exists a unique homomorphism $h : R_S \to R'$ such that $j' = h \circ j$.
\end{itemize}
With this definition it is apparent that $R_S$ is determined up to a canonical isomorphism. For a~simple proof of existence of $R_S$, see \cite[Proposition 9.2]{Lam}. Unfortunately in general it is difficult to work with $R_S$. Its elements are linear combinations of words $j(r_1) j(s_1)^{-1} j(r_2) j(s_2)^{-1} \cdots$, and it is not easy to describe relations between such words.

If $R$ was assumed to be commutative, we could combine all denominators in any single word and find a common denominator for any finite collection of words, so every element of $R_S$ would be of the form $j(r) j(s)^{-1}$. The same conclusion can be made under a weaker assumption, called the \textbf{right Ore condition}: for every $r \in R$ and $s \in S$ there exist $r' \in R$ and $s' \in S$ such that $rs'=sr'$. There is an analogous (but non-equivalent) \textbf{left Ore condition}, which implies that every element of $R_S$ is of the form $j(s)^{-1}j(r)$. 

Concerning the kernel of $j$, first note that $\Ker(j)$ is an ideal. Now consider the sets
\begin{equation}
    I_1 = \{ r \in R \, | \, \exists s \in S \ rs = 0 \}, \qquad I_2 = \{ r \in R \, | \, \exists s \in S \ sr = 0 \}.
\end{equation}
Clearly $I_1, I_2 \subset \Ker(j)$. It is easy to check that the right Ore condition implies that $I_1 $ is an ideal (and similarly, the left Ore condition implies that $I_2$ is an ideal). Therefore if we supplement the right Ore condition with the \textbf{right cancellability condition} $I_2 \subset I_1$, then $I_1$ is a natural candidate for $\Ker(j)$. This guess turns out to be correct:  
\begin{proposition}
The following conditions are equivalent:
\begin{itemize}
    \item $R,S$ satisfy the right Ore condition and the right cancellability condition,
    \item every element of $R_S$ is of the form $j(r)j(s)^{-1}$ and $\Ker(j) = I_1$.
\end{itemize}
The same is through if we replace right by left and $I_1$ by $I_2$ throughout. 
\end{proposition}

We refer to \cite[Theorem 10.6 and Corollary 10.11]{Lam} for a proof. Its main part is an explicit (but tedious) construction of $R_S$ as a quotient set of $R \times S$, with the class of pair $(r,s)$ corresponding to a "fraction" $rs^{-1}$.

Now let $\mathfrak g$ be a Lie algebra and let $S_0 \subset \mathfrak g$ be a set consisting of locally ad-nilpotent elements, i.e.\ for every $x \in S_0$ and $y \in \mathfrak g$ there exists $n \in \N$ such that $[x , \cdot]^n(y)=0$. Let $R = \U \mathfrak g$ and let $S \subset R$ be the multiplicative set generated by $S_0$, i.e.\ all element in $R$ that can be expressed as words in letters from the alphabet $S_0$. $R$ has no zero divisors, so left and right cancellability conditions are satisfied. We prove that the right Ore condition is satisfied. By the same method, the left Ore condition also holds. It is sufficient to show that if $r \in R$ and $s \in S_0$, there exist $r' \in R$ and $s' \in S$ such that $rs' = sr'$. Indeed (since $r$ is a polynomial in elements of $\mathfrak g$), by assumption on $S_0$ there exists $n \in \N$ such that $[s, \cdot]^n (r)=0$, and we may take $s'=s^{n}$:
\begin{equation}
    rs^{n} = s \sum_{i=0}^{n-1} \binom{n}{i} s^{n-i-1} [\cdot, s]^{i}(r).  
\end{equation}

In summary, $\U \mathfrak g$ is canonically embedded in the algebra $R_S$, in which elements of $S_0$ are invertible and (upon identifying $\U \mathfrak g$ with its image in $R_S$) every element can be written as $rs^{-1}$ and as $s'^{-1} r'$ for some $r,r' \in \U \mathfrak g$ and $s,s' \in S$.

\section*{Declarations}
\subsubsection*{Funding and/or Conflicts of interests/Competing interests}

The authors have no competing interests to declare that are relevant to the content of this article.

\end{document}